\documentclass[11pt]{article}

\usepackage[margin= 1.25 in]{geometry}

\usepackage[ngerman, english]{babel}
\usepackage{amsmath,amsfonts,amssymb,amsthm}
\usepackage{graphicx}
\usepackage{pstricks}
\usepackage{bbm}
\usepackage[font=small,labelfont=bf]{caption}
\usepackage{eqnarray}
\usepackage[onehalfspacing]{setspace}
\usepackage[mathlines]{lineno}
\usepackage{mathtools}

\newcommand{\ignore}[1]{}

\allowdisplaybreaks

\newcommand{\keywords}[1]{\textbf{Keywords:}\quad #1}

\def\a{\alpha}

\def\d{\delta}

\def\lam{\lambda}
\def\g{\gamma}

\def\1{{1\hskip-2.5pt{\rm l}}}

\def\A{\mathcal A}

\def\d{\delta}
\def\f{\varphi}
\def\g{\gamma}
\def\R{\mathbb R}

\long\def\ignore#1{}

\newcommand{\LV}{\left\vert}
\newcommand{\RV}{\right\vert}

\newcommand{\myqed}{\hfill$\square$}

\DeclareMathOperator*{\argmax}{arg\,max}
\DeclareMathOperator*{\argmin}{arg\,min}

\theoremstyle{definition}
\newtheorem{definition}{Definition}[section]
\newtheorem{example}[definition]{Example}

\theoremstyle{plain}
\newtheorem{proposition}[definition]{Proposition}
\newtheorem{lemma}[definition]{Lemma}
\newtheorem{corollary}[definition]{Corollary}
\newtheorem{theorem}[definition]{Theorem}

\usepackage[backend=biber, maxbibnames=9, maxcitenames=2,uniquelist=false, citestyle=authoryear, bibstyle=authoryear, style=authoryear-comp, sorting= nyt, block =space, hyperref=true,  natbib=true, giveninits=true]{biblatex}
\defbibheading{head}{}
\usepackage{csquotes}

\newcommand*\patchAmsMathEnvironmentForLineno[1]{%
	\expandafter\let\csname old#1\expandafter\endcsname\csname #1\endcsname
	\expandafter\let\csname oldend#1\expandafter\endcsname\csname end#1\endcsname
	\renewenvironment{#1}%
	{\linenomath\csname old#1\endcsname}%
	{\csname oldend#1\endcsname\endlinenomath}}%
\newcommand*\patchBothAmsMathEnvironmentsForLineno[1]{%
	\patchAmsMathEnvironmentForLineno{#1}%
	\patchAmsMathEnvironmentForLineno{#1*}}%
\AtBeginDocument{%
	\patchBothAmsMathEnvironmentsForLineno{equation}%
	\patchBothAmsMathEnvironmentsForLineno{align}%
	\patchBothAmsMathEnvironmentsForLineno{flalign}%
	\patchBothAmsMathEnvironmentsForLineno{alignat}%
	\patchBothAmsMathEnvironmentsForLineno{gather}%
	\patchBothAmsMathEnvironmentsForLineno{multline}%
}

\title{A Taste for Variety\thanks{This research has been funded by the Deutsche Forschungsgemeinschaft (DFG, German Research Foundation) -- Project Number 461570745.}}
\author{Galit Ashkenazi-Golan\thanks{Department of Mathematics, London School of Economics} \and
	  Dominik Karos\thanks{Center for Mathematical Economics, Bielefeld University} \and Ehud Lehrer\thanks{School of Mathematical Sciences, Tel Aviv University and Durham University, Durham, UK.}}

\begin{document}
	
	
	\maketitle

	\begin{abstract}
		A decision maker repeatedly chooses one of a finite set of actions. In each period, the decision maker's payoff depends on fixed basic payoff of the chosen action and the frequency with which the action has been chosen in the past. We analyze optimal strategies associated with three types of evaluations of infinite payoffs: discounted present value, the limit inferior, and the limit superior of the partial averages. We show that when the first two are the evaluation schemes, a stationary strategy can always achieve the best possible outcome. However, for the latter evaluation scheme, a stationary strategy can achieve the best outcome only if all actions that are chosen with strictly positive frequency by an optimal stationary strategy have the same basic payoff.
	\end{abstract}
	\noindent \keywords{Repeated decision problem; intertemporal choice; time-inconsistent preferences; habit formation}
	\medskip
	
	

	\section{Introduction}
	
	When Phil Connors\footnote{%
		Played by Bill Murray in ``Groundhog Day'', 1993}
	was trapped in a time loop, he initially enjoyed being able to do as he liked without fearing any repercussions.
	Yet, after a while, he became depressed as the rather limited entertainment options available in Punxsutawney did not measure up to his taste for variety.
	In this paper we investigate what Phil's optimal long-term payoff would have been, had he not been able to escape his temporal prison.
	That is, we consider a decision maker who has to repeatedly choose from a finite set of actions and whose stage payoff depends both on the action itself and also on how often she has chosen it in the past.
	
	The model that we propose here looks rather innocuous.
	There is a finite set of actions, each endowed with a fixed basic payoff, and at each period the decision maker has to choose one of them.
	Her stage utility from choosing some action $a$ is $a$'s basic payoff multiplied by a factor that depends on the frequency with which $a$ has been played so far and her taste for variety.
	The greater this frequency, the smaller the utility.
	
	The decision maker is interested in her long-run payoff. We analyze three types of long-term payoff evaluations: the limit inferior and limit superior of the partial averages and the discounted one.
	It turns out that the limit inferior and discounted evaluations share the important feature that their optimal outcomes can be achieved by stationary strategies. However, the optimal strategy for the limit superior evaluation is stationary only in the degenerate case where all actions chosen with strictly positive frequency by an optimal stationary strategy have the same basic payoff.

	
	Classical economic theory assumes static preferences and discounted utility, as proposed by \citet{samuelson1937} and later motivated with an axiomatic foundation by \citet{koopmans1960stationary}.
	Since then, this approach has been challenged in various contexts.
	Arguably, the most developed one is choice under uncertainty.
	Based on the famous example of \citet{allais1953} dynamic consistency was challenged, and two major branches of the literature emerged:
	one focussed on behavioral aspects and challenged expected utility as a whole \citep[e.g.,][]{thaler1981, machina1989}; the other focussed on optimizing stage decisions based on one's experience from the past \citep{gilboa1995case}.
	This form of ``instance-based learning'', which has also found its way into cognitive science \citep{gonzales_lerch_lebiere2003, stewart_chater_brown2006}, asserts a causal connection between past and present behavior rather than dynamic consistency.
	Yet, this paper does not cover dynamic inconsistencies that originate in uncertainty, so we refrain from providing an extensive overview of the literature here and refer to \citet{etner_jeleva_tallon2012}.
	
	In this paper, we exclusively cover complete information.
	In this case as well, several forms of time-inconsistent behavior are present: first, there is the classical present bias in which agents over-discount future payoffs \citep{odonoghue_rabin1999}.
	As we will investigate how the past (rather than the future) affects current decisions, this is not the behavior we are interested in.
	Somewhat closer in spirit is the research on reference-dependent utility \citep{kHoszegi2006model, odonoghue_sprenger2018} if the reference point is based on the past \citep{baucells_weber_welfens2011}.
	However, our decision maker does not derive a reference point based on past choices, but rather obtains (or loses) some utility for making the same choice very often.
	
	Our decision maker's preferences are more closely related to the idea of ``habit formation'' and to the model of \citet{kaiser_schwabe2012}.
		Originally, \citet{becker1988theory} propose a model of ``rational addiction'' in which a decision maker maximizes aggregated future utility whereby the stage utility at any time depends on past consumption.
	In this flavor, axiomatic characterizations of history-dependent consumer preferences over future consumption paths were developed to account for this effect \citep[e.g.,][]{rozen2010foundations, he_dyer_butler2013, rustichini_siconolfi2014}.
	These models play a crucial role in macroeconomic models as they explain some phenomena and fit data better than standard expected utility theory.
	For instance, \citet{boldrin_christiano_fisher2001} introduce habit persistence into a business cycle model, and \citet{constantinides1990} uses habit persistence to resolve the equity premium puzzle \citep[cf.][]{mehra_prescott1985}.

	Outside the scope of economic theory, a similar idea has been brought forward in psychology.
	The ``mere exposure effect'' \citep{zajonc1968}, also called the ``familiarity effect'', describes the change in preferences from simply being exposed to some object.
	Originally, only positive effects were observed in experiments: an object became more popular as the decision maker was exposed to it more often.
	But there are scenarios where this effect is reversed \citep{crisp_hutter_young2008}, or the relation is even non-monotonic: increasing, reaching a satiation point, and decreasing again as exposure increases \citep{zajonc_shaver_tavris_vankreveld1972, williams1987}.
	In particular, research on the interdependences between the mere exposure effect and boredom \citep{bornstein_kale_cornell1990} or the novelty principle \citep{liao_yeh_shimojo2011} has provided a range of stage preferences over objects that depend on past exposure.
	
	The paper is organized as follows.
	In Section \ref{sec:preliminaries} we introduce the necessary notation and provide some examples that highlight the different ways an infinite history of actions might be evaluated.
	In particular, we illustrate by means of an example with two actions that the optimal limes superior cannot be achieved by a stationary strategy.
	In Section \ref{sec:greedy_and_stat} we investigate greedy histories, which maximize the stage utility in each period.
	We observe that such strategies are stationary, but we show that they are far from optimal even within the set of stationary strategies.
	In Section \ref{sec:liminf} we show that the optimal limit inferior can be achieved by a stationary strategy.
	Moreover, we show that the action frequencies of optimal histories are first-order stochastically ordered as the fatigue factor increases: the larger this factor, the more weight the optimal frequency will put on poor actions.
	Section \ref{sec:limsup} deals with the optimal limit superior.
	We show that the sequence of optimal average payoffs after finite time converges against the optimal limit superior and we use this observation to show that the latter cannot be achieved by a stationary strategy unless the optimal stationary strategy chooses the same action in each period.
    Section \ref{sec: discounting} deals with two aspects of discounting: discounting future payoffs and discounting the effect of past uses of actions.
    Discounting future payoffs means that one values future positive payoffs less than present ones. This is because one prefers to have good things now rather than later.
    Discounting the effect of past uses of actions means the impact of past experience on the present utility diminishes with time.
    For example, if one eats the same meal every day, one will eventually get tired of it.
    However, if one had a delicious meal yesterday, he or she would prefer the same meal today less than if he or she had it only a year ago.
    The main result of this section 
    states that the optimal outcome for a relatively patient decision maker can be obtained with stationary strategies.

	\section{Preliminaries}\label{sec:preliminaries}
	
	Let $A$ be a finite set of \emph{actions} that a decision maker has to choose from at each period $t\in\mathbb N\setminus\{0\}$ and let $u: A\to [0,\infty)$ be the decision maker's \emph{basic} payoff function.
	A \emph{finite history of length $T$} is a map $\vec a:\left\{1,\ldots,T\right\}\rightarrow A$, and an \emph{infinite history} is a map $\vec a:\mathbb N\setminus\{0\}\rightarrow A$.
	For $T\in\mathbb N$ we denote the set of histories of length $T$ by $A^T$, where $A^0$ only contains the empty history.
	The set of all finite histories is denoted by $A^{<\infty}$, that is, $A^{<\infty}=\bigcup_{T=0}^{\infty}A^T$, and the set of all infinite histories is denoted by $A^{\infty}$.
	For an infinite history $\vec a\in A^{\infty}$ and a non-negative integer $t\in\mathbb N\setminus\{0\}$ we write $\vec a_t$ for the $t$-th element of the sequences, $\vec a^{t}$ for the finite history $\left(\vec a_1,\vec a_2,\ldots,\vec a_t\right)$, and also  $\vec{a}_0=\vec a^0=\emptyset$.
	A \emph{strategy} is a map $\sigma:A^{<\infty}\rightarrow A$.
	
	We denote the indicator function by $\mathbbm 1$, that is, for a history $\vec a$ we have that $\1_{\vec a_s=a}=1$ if $\vec a_s=a$ and $\1_{\vec a_s=a}=0$ otherwise.
	We define the map $\f:A\times A^{<\infty}\rightarrow\Delta(A)$ as
	\begin{align*}
		\f\left(a\middle\vert \vec a^t\right) =\begin{cases}
			\frac{1}{t}\sum_{s=1}^{t}\1_{\vec a_s=a}, & \text{if } t\geq 1,\\
			0, & \text{if } t=0.
		\end{cases}
	\end{align*}
	That is,  $\f\left(a\middle\vert \vec a^{t-1}\right)$ is the \emph{frequency} of $a$ in the history $\vec a^{t-1}=\left(\vec a_1,\vec a_2,\ldots,\vec a_{t-1}\right)$.
	
	In the repeated decision problem the decision maker experiences some ``fatigue'' when choosing the same action repeatedly.
	More precisely, there is $\gamma\in\left(0,1\right]$ such that when taking action $a\in A$ after history $\vec a^{t-1}$, the \emph{stage payoff} at stage $t$ is
	\begin{align*}
		u_{\g,t}(a;\vec a^{t-1})=\left(1-\gamma\f\left(a\middle\vert \vec a^{t-1}\right)\right)u(a_t).
	\end{align*}
	A large $\gamma$ represents strong fatigue or a strong ``taste for variety'': the stage payoff quickly declines if an action is chosen repeatedly.
	If $\gamma=0$, there is no need for variety, and the maximization of stage payoff and basic payoff are equivalent.
	We exclude this case.
	
	We are interested in the ``maximal'' payoff a decision maker can obtain in such a repeated decision problem.
	Specifically, for an infinite history $\vec{a}\in A^{\infty}$ the decision maker's \emph{average} (\emph{undiscounted}) \emph{utility} at $T$ is
	\begin{align*}
		U^T_{\gamma}\left(\vec{a}\right)=\frac{1 }{T}\sum_{t=1}^{T} u_{\g,t}(a;\vec a^{t-1})=\frac{1 }{T}\sum_{t=1}^{T}\left(1-\gamma\f\left(a_t\middle\vert a^{t-1}\right)\right)u\left(a_t\right).
	\end{align*}
	Surely, $U^T_{\gamma}\left(\vec{a}\right)<\infty$ for all $\vec a\in A^{\infty}$ and all $T\in\mathbb N\setminus\{0\}$.
	Yet, in general, the sequence $\left(U^T_{\gamma}\left(\vec{a}\right)\right)_{T\in\mathbb N\setminus\{0\}}$ will not converge.
	
	\begin{example}\label{exa:1_2_sequence_initial}
		Let $A=\left\{a,b\right\}$ with $u(a)=1$ and $u(b)=10$.
		Consider the history $\vec a$ that is defined by $\vec a_1=a$, $\vec a_2=b$, $\vec a_3=a$ and
		\begin{align*}
			\vec a_t = \begin{cases}
				a, & \text{if there is an odd } m\in\mathbb N\setminus\{0\} \text{ such that } 3\cdot 2^m+1\leq t \leq 3\cdot 2^{m+1},\\
				b, & \text{if there is an even } m\in\mathbb N\setminus\{0\} \text{ such that } 3\cdot 2^m+1\leq t \leq 3\cdot 2^{m+1},
			\end{cases}
		\end{align*}
		for $t\geq 4$.
		That is, $\vec a=\left(a,b,a,b,b,b,a,a,a,a,a,a,b,\ldots\right)$.
		In this sequence, exponentially increasing blocks of consecutive $a$'s and $b$'s are played alternating.
		In particular, from $t\geq 4$ onwards each block is as long as the entire history before the block, so that the frequency of either action fluctuates between $1/3$ at the beginning of each block and $2/3$ at the end.
		The sequence of average utilities of this infinite history does not converge.
		Intuitively, it will be lowest at the end of any $a$-block, and highest at the end of any $b$-block.
		We shall have a closer look at this behavior later.
		\myqed
	\end{example}

	%
	%

	\noindent As the sequence $\left(U^T_{\gamma}\left(\vec{a}\right)\right)_{T\in\mathbb N\setminus\{0\}}$ might not converge for all $\vec a\in A^{\infty}$, there is no ``obvious'' way to compare two infinite histories $\vec a,\vec b\in A^{\infty}$.
	Yet, as every sequence of average utility is bounded, we can use their upper and lower accumulation points for comparisons.
	To keep notation short, define for any $\vec a\in A^{\infty}$
	\begin{align*}
		\overline{V}_{\gamma}\left(\vec{a}\right) &= \limsup_{T\rightarrow\infty} U^T_{\gamma}\left(\vec{a}\right) & & \text{and} &
		\underline{V}_{\gamma}\left(\vec{a}\right) &= \liminf_{T\rightarrow\infty} U^T_{\gamma}\left(\vec{a}\right),
	\end{align*}	
	which are the highest and lowest accumulation points that the sequence of average utilities can reach for the history $\vec a$.
	Moreover, let
	\begin{align*}
		\overline V_{\gamma} &= \sup\left\{\overline V_{\gamma}\left(\vec a\right)\middle\vert \vec a\in A^{\infty}\right\}
		& & \text{and} &
		\underline V_{\gamma} &= \sup\left\{\underline V_{\gamma}\left(\vec a\right)\middle\vert \vec a\in A^{\infty}\right\}.
	\end{align*}
	Thus, for each $v<\overline V_{\gamma}$ there is a history $\vec a$ whose average utility is at least $v$ in infinitely many periods.
	Likewise, for each $v<\underline V_{\gamma}$ there is a history $\vec a$ whose average utility is at least $v$ in all but finitely many periods.
	
	\begin{example}\label{exa:limsup}
		Recall the history $\vec a$ from Example \ref{exa:1_2_sequence_initial}.
		As $t$ gets large, the average frequency of the action that is played in a block is approximated by
		\begin{align}\label{equ:exa_freq}
			x = \int_{0}^1\left( 1-\frac{2}{3s+3}\right)ds = 1-\left(\frac{2}{3}\ln\left(2\right)-\ln\left(1\right)\right) = 1 - \frac{2}{3}\ln(2).
		\end{align}
		Thus, even though the frequencies of $a$ and $b$ do not converge, the average of $\varphi\left(a\middle\vert a^{t-1}\right)$ taken over all $t$ with $\vec a_t=a$ converges towards $x$, and the same is true for the average of $\varphi\left(b\middle\vert a^{t-1}\right)$ taken over all $t$ with $\vec a_t=b$.
		Hence, at the end of any block of $a$'s, the average payoff is approximately
		\begin{align*}
			U_{\gamma}^T\left(\vec a\right) &\approx \frac{2}{3}\left(1-\gamma\left(1 - \frac{2}{3}\ln(2)\right) \right) u(a) + \frac{1}{3}\left(1-\gamma\left(1 - \frac{2}{3}\ln(2)\right) \right) u(b)\\
			&= 4\left(1-\gamma\left(1 - \frac{2}{3}\ln(2)\right) \right).
		\end{align*}
		Observe that for such $T$ it holds that
		\begin{align*}
			u^T\left(a;\vec a\right) &= \left(1-\frac{2}{3}\gamma\right)\leq U_{\gamma}^T\left(\vec a\right) & \text{and} & &
            u^{T+1}\left(b;\vec a\right) &= 10\left(1-\frac{1}{3}\gamma\right)\geq U_{\gamma}^T\left(\vec a\right),
		\end{align*}
		for all $\gamma\in\left[0,1\right]$.
		Thus, $U_{\gamma}^T\left(\vec a\right)$ is minimized at the end of each $a$-block, and we find $\underline V_{\gamma}\left(\vec a\right) = 4\left(1-\gamma\left(1 - \frac{2}{3}\ln(2)\right) \right)$.
		
		In order to find $\overline V_{\gamma}\left(\vec a\right)$ we show that $U_{\gamma}^T\left(\vec a\right)$ achieves its maxima always at the end of $b$-blocks.
		So, consider a (large) $b$-block.
		We want to show that the average utility of $\vec a$ is increasing throughout the entire block.
		So, keeping in mind that the block is large, let $x\in\left[0,1\right]$ and consider the period after a fraction $x$ of the block has passed.
		The frequencies of $a$ and $b$ at this point in time are given by $f_a \approx \frac{2}{3x+3}$ and $f_b=1-f_a\approx\frac{3x+1}{3x+3}$.
		Hence, the stage utility is given by
		\begin{align*}
			v(x)=10\left(1-\gamma f_b\right) \approx \frac{30\left(1-\gamma\right)x + 30-10\gamma}{3x+3}.
		\end{align*}
		The average frequency of $a$ at $x$ (taken over the periods where $a$ has been chosen) is still given in (\ref{equ:exa_freq}).
		The average frequency of $b$ at $x>0$ is given by
		\begin{align*}
			\frac{1}{x}\int_{0}^x 1-\frac{2}{3s+3}ds
			= 1-\frac{1}{x}\left(\frac{2}{3}\ln\left(3x+3\right)-\frac{2}{3}\ln\left(3\right)\right)
			= 1- \frac{2}{3x}\ln\left(x+1\right).
		\end{align*}
		Thus, the average utility at $x$ is approximated by
		\begin{align*}
			U(x) &= f_a \left(1-\gamma\left(1 - \frac{2}{3}\ln(2)\right) \right) u(a) + f_b \left(1-\gamma\left(1- \frac{2}{3x}\ln\left(x+1\right)\right) \right)u(b)\\
			&= \frac{2}{3x+3}\left(1-\gamma\left(1 - \frac{2}{3}\ln(2)\right) \right) + \frac{3x+1}{3x+3} \left(1-\gamma\left(1- \frac{2}{3x}\ln\left(x+1\right)\right) \right)10.
		\end{align*}
		In particular, $U(1)<v(1)$ for all $\gamma\in\left[0,1\right]$.
		As $U$ is increasing at $x$ if and only if $v(x)>U(x)$, and $v$ is falling in $x$, this implies that $U$ reaches its maximum at $x=1$.
		Thus, we obtain
		\begin{align*}
			\overline V_{\gamma}\left(\vec a\right) = U(1) = 7\left(1-\gamma\left(1 - \frac{2}{3}\ln(2)\right) \right)
		\end{align*}
		for the highest limit point that $U^T_{\gamma}\left(\vec a\right)$ can reach.
		\myqed
	\end{example}

	\section{Greedy behavior and stationary strategies}\label{sec:greedy_and_stat}
	
	A simple strategy $\sigma$ that a decision maker might follow is to maximize her stage utility at each $t$, that is, choose her action at $t$ according to
	\begin{align*}
		\vec a_t = \sigma\left(\vec a^{t-1}\right) \in \argmax_{a\in A} \left(1-\gamma\varphi\left(a\middle\vert \vec a^{t-1}\right)\right)u(a).
	\end{align*}
	We call such a strategy a \emph{greedy strategy}.
	In this case the frequency $\varphi\left(a\middle\vert \vec a^{t}\right)$ necessarily converges for all $a\in A$.
	
	\begin{proposition}\label{pro:greedy}
		Let $\vec a\in A^{\infty}$ be the history evolving from a greedy strategy.
		Then $\varphi\left(a\middle\vert \vec a^{t}\right)$ converges for all $a\in A$ and
		\begin{align}\label{equ:frequency_greedy}
			\lim_{t\rightarrow\infty}\varphi\left(a\middle\vert \vec a^{t}\right) = \frac{\displaystyle \gamma - \LV A^*\RV + \sum_{b\in A^*}\frac{u(a)}{u(b)}}{\displaystyle\gamma \sum_{b\in A^*}\frac{u(a)}{u(b)}}
		\end{align}
		for all $a\in A^*$, where $A^*$ is the set of actions that are chosen infinitely often.
	\end{proposition}
	
	\begin{proof}
		For each $\varepsilon>0$ there is $T\in\mathbb N\setminus\{0\}$ such that
		\begin{align*}
			\LV\left(1-\gamma\varphi\left(a\middle\vert \vec a^{t-1}\right)\right)u(a) - \left(1-\gamma\varphi\left(b\middle\vert \vec a^{t-1}\right)\right)u(b)\RV <\varepsilon
		\end{align*}
		for all $a,b\in A^*$ and all $t\geq T$.
		As $\sum_{a\in A^*}\varphi\left(a\middle\vert \vec a^{t}\right)=1$ for all $t\geq 1$, the frequencies converge.
		Let $f_a = \lim_{t\rightarrow\infty}\varphi\left(a\middle\vert \vec a^{t}\right)$.
		Then $\left(1-\gamma f_a\right)u(a)=\left(1-\gamma f_b\right)u(b)$ for all $a,b\in A^*$.
		Solving for $b$ and summing over all $b$ we find that
		\begin{align*}
			1=\sum_{b\in A^*} f_b = \frac{1}{\gamma}\sum_{b\in A^*} \left(1-\frac{u(a)}{u(b)}\left(1-\gamma f_a\right)\right) = \frac{1}{\gamma}\left(\LV A^*\RV - \left(1-\gamma f_a\right)\sum_{b\in A^*}\frac{u(a)}{u(b)}\right)
		\end{align*}
		Solving for $f_a$ delivers (\ref{equ:frequency_greedy}).
	\end{proof}
	
	\noindent The expression in (\ref{equ:frequency_greedy}) provides a bound on the number of actions that can be played with positive probability.
	In particular, for $\gamma<1$ it is possible that the greedy strategy will only choose a single action that is played at every $t$.
	
	As seen in Example \ref{exa:1_2_sequence_initial}, frequencies do not converge for all $\vec a\in A^{\infty}$.
	Yet, if they do, as for the greedy strategy above, the average utility converges as well.
	We say that a history $\vec a\in A^{\infty}$ is \emph{stationary} if $\lim_{t\rightarrow\infty}\varphi\left(a\middle\vert \vec a^{t-1}\right)$ exists for all $a\in A$.
	In this case we write $\varphi\left(a\middle\vert\vec a\right) = \lim_{t\rightarrow\infty}\varphi\left(a\middle\vert \vec a^{t-1}\right)$.
	If there is no risk of confusion, we will even write $\varphi(a)=\varphi\left(a\middle\vert\vec a\right)$.
	The limit of the average utilities is then given by
	\begin{align}\label{equ:stationary_limit}
		\overline{V}_{\gamma}\left(\vec{a}\right)  = \underline{V}_{\gamma}\left(\vec{a}\right) = \lim_{T\rightarrow\infty}U^T_{\gamma}\left(\vec a\right) = \sum_{a\in A}\varphi\left(a\right)\left(1-\gamma\varphi\left(a\right)\right)u(a).
	\end{align}
	We denote the optimal limit that can achieved by any stationary history by
	\begin{align*}
		V_{\gamma}^* &= \sup\left\{\underline V\left(\vec a\right)\mid \vec a\in A^{\infty} \ \text{is stationary} \right\}.
	\end{align*}
    Finally, we say that a strategy is  \emph{stationary} if it generates a stationary history.

	\begin{example}\label{exa:greedy}
		Let $A=\left\{a,b\right\}$ with $u(a)=1$ and $u(b)=10$.
		If $\gamma\leq 0.9$, the greedy strategy will choose $b$ for all $t$.
		If $\gamma>0.9$, then the frequencies achieved by the greedy strategy are $\varphi\left(a\right) = \frac{10\gamma-9}{11\gamma}$ and $\varphi\left(b\right) = \frac{\gamma +9}{11\gamma}$.
		Thus,
		\begin{align*}
			\underline V_{\gamma}\left(\vec a\right) = \frac{10\gamma -9}{11\gamma}\left(1-\gamma \frac{10\gamma -9}{11\gamma}\right)u(a) + \frac{\gamma +9}{11\gamma}\left(1- \gamma\frac{\gamma +9}{11\gamma}\right)u(b) = \frac{20-10\gamma}{11}.
		\end{align*}
		In particular, for $\gamma=0.9$, only $b$ will be chosen and its stage payoff converges towards 1.
		\myqed
	\end{example}
	
	\noindent The previous example illustrates that the greedy strategy does not deliver particularly high payoffs.
	Indeed, the ``good'' actions are overused so that their stage payoffs become very low, resulting in a low average payoff.
	Finding $V^*_{\gamma}$ is indeed not very difficult; by (\ref{equ:stationary_limit}), it is given by
	\begin{align}\label{equ:max_problem_stationary}
		V_{\gamma}^* = \max_{x\in\Delta(A)}\sum_{a\in A}x_a\left(1-\gamma x_a\right)u(a),
	\end{align}
	where $\Delta(A)$ denotes the set of probability measures over $A$.
	As the objective function is strictly quasi-concave for all $\gamma>0$, the maximization problem in (\ref{equ:max_problem_stationary}) has a unique solution $x^*\in\Delta(A)$.
	In particular, every stationary history $\vec a$ with $\varphi\left(\cdot\middle\vert\vec a\right)=x^*$ is optimal.
	The next proposition specifies these optimal frequencies.
	
	\begin{proposition}\label{pro:stat}
		Let $\vec a\in A^{\infty}$ be the history evolving from an optimal stationary strategy.
		Then
		\begin{align}\label{equ:frequency_stat}
			\varphi\left(a\right) = \frac{\displaystyle 2\gamma - \LV A^*\RV + \sum_{b\in A^*}\frac{u(a)}{u(b)}}{\displaystyle 2\gamma \sum_{b\in A^*}\frac{u(a)}{u(b)}}
		\end{align}
		for all $a\in A^*$, where $A^*\subseteq A$ is the set of actions with $\varphi(a)>0$.
	\end{proposition}
	
	\begin{proof}
		The first-order conditions of the maximization problem in (\ref{equ:max_problem_stationary}) are
		\begin{align*}
			\left(1-2\gamma\varphi\left(a\right)\right)u(a)=\left(1-2\gamma\varphi\left(b\right)\right)u(b).
		\end{align*}
		for all $a,b\in A^*$.
		With the same steps as in the proof of Proposition \ref{pro:greedy} one obtains (\ref{equ:frequency_stat}).
	\end{proof}

	\begin{example}\label{exa:opt_stat}
		Let $A=\left\{a,b\right\}$ with $u(a)=1$ and $u(b)=10$.
		Let $\vec a$ be the history evolving from an optimal stationary strategy.
		For $\gamma\leq\frac{9}{20}$ action $a$ will not be played with positive probability.
		For $\gamma>\frac{9}{20}$, the optimal frequencies are $\varphi(a) = \frac{20\gamma -9}{22\gamma}$ and $\varphi(b) = \frac{2\gamma +9}{22\gamma}$.
		Thus,
        \small
		\begin{align*}
			\underline V_{\gamma}\left(\vec a\right) = \frac{20\gamma -9}{22\gamma}\left(1-\gamma \frac{20\gamma -9}{22\gamma}\right)u(a) + \frac{2\gamma +9}{22\gamma}\left(1- \gamma\frac{2\gamma +9}{22\gamma}\right)u(b) = \frac{-40\gamma^2+80\gamma+81}{44\gamma}.
		\end{align*}
		\normalsize
        In particular, this expression is strictly larger than the average utility of the greedy strategy in Example \ref{exa:greedy}.
		\myqed
	\end{example}
	
	\noindent A special case of optimal stationary histories emerges if $A$ contains exactly two elements and $\gamma =1$.
	In this case Proposition \ref{pro:stat} immediately implies the following corollary.
	
	\begin{corollary}
		Let $\gamma=1$, let $A=\left\{a,b\right\}$, and let $\vec a\in A^{\infty}$ be an optimal stationary history.
		Then $\varphi(a)=\varphi(b)=\frac{1}{2}$.
		In particular, $V_1^*=\frac{1}{4}\left(u(a)+u(b)\right)$.
	\end{corollary}

	\noindent At this point it has become clear that defining what an optimal strategy is crucially depends on how the evolving histories are evaluated.
	Finding optimal stationary strategies is rather simple, as shown in Proposition \ref{pro:stat}, yet Examples \ref{exa:limsup} and \ref{exa:opt_stat} illustrate that stationary strategies might not be able to achieve $\overline V_{\gamma}$ as an average utility.
	Indeed, for $\gamma=1$, the strategy in Example \ref{exa:limsup} achieves an average utility of $\frac{14}{3}\ln(2)\approx 3.23$, while the best stationary strategy in Example \ref{exa:opt_stat} achieves only $\frac{11}{4}=2.75$.
	
	In the remainder of the paper we shall investigate how the three possible values, that is, the optimal highest accumulation point, the optimal lowest accumulation point, and the optimal limit (if it exists) compare.
	They must satisfy
	\begin{align*}
		V^*_{\gamma} \leq \underline V_{\gamma} \leq \overline V_{\gamma}.
	\end{align*}
	Our main results will be that here the first inequality is actually an equality, while the second inequality is strict if there are at least two actions $a,b\in A$ with $u(a)\neq u(b)$ that are chosen with positive frequency in an optimal stationary history.

	\section{Stationary strategies achieve $\underline V_{\gamma}$}\label{sec:liminf}
	
	Finding a strategy such that the evolving history $\vec a$ that achieves $\underline V_{\gamma}\left(\vec a\right)=\underline V_{\gamma}$ is essentially a dynamic program on a countable state space.
	Unfortunately, these problems typically lack a tractable structure, so there are no general results that could be helpful in the current context.
	Thus, we will have to develop some tools to obtain our result in Subsection \ref{subsec:opt_V*}.
	In Subsection \ref{subsec:fatique} we shall then investigate how the optimal frequencies change as the parameter $\gamma$ varies.

	\subsection{The optimality of $V_{\gamma}^*$}\label{subsec:opt_V*}
	
	Let $\vec a$ be a history.
	For any $t_1,t_2\in\mathbb N\setminus\{0\}$ with $t_2>t_1$ let the \emph{block} from $t_1$ to $t_2$ in $\vec a$ be the sequence of actions $\left(\vec a_{t_1+1},\vec a_{t_1+2},\ldots,\vec a_{t_2}\right)$.
	The average utility within this block is given by
	\begin{align}\label{equ:W}
		W_{\gamma}=W_{\gamma}\left(\vec a,t_1,t_2\right)=\frac{1}{t_2-t_1}\sum_{s={t_1+1}}^{t_2}\left(1-\g\varphi\left(a_s\mid \vec a^{s-1}\right)\right)u(a_s).
	\end{align}
	Let $p(a)$ be the frequency with which $a$ is played in the block, that is,
	\begin{align}\label{equ:p}
		p(a) = p\left(a;\vec a,t_1,t_2\right) =\frac{1}{t_2-t_1}\sum_{s=t_1+1}^{t_2}\mathbbm 1_{\vec a_s = a}.
	\end{align}
	If such a block is ``not too long'', the frequencies will not change much between $t_1$ and $t_2$.
	We want to use this observation to derive an approximation of $W_{\gamma}$ by means of $p$ and the frequency at the beginning, that is, $\varphi\left(\cdot\mid\vec a^{t_1}\right)$.
	In particular, we define  $\widetilde U_\gamma$ as
	\begin{align}\label{equ:approx_U2}
		\widetilde{U}_{\gamma}=\widetilde{U}_{\gamma}\left(\vec a,t_1,t_2\right)=\sum_{a\in A}p(a)\left(1-\g\varphi\left(a\middle\vert\vec a^{t_1}\right)\right)u(a).
	\end{align}
	We show that $\widetilde U_{\gamma}$ is close to $W_{\gamma}$ if $t_2$ is relatively close to $t_1$, that is, if $\frac{t_2-t_1}{t_1}$ is small.

	\begin{lemma}\label{lem:w_U_approximation}
		Let $\vec a\in A^{\infty}$ be a history and let $t_1,t_2\in\mathbb N\setminus\{0\}$ with $t_2>t_1$.
		Then
		\begin{align}\label{equ:w_U_approximation}
			\big| W_{\gamma}- \widetilde{U}_{\gamma}\big|\le 2\frac{t_2-t_1}{t_1} \g \sum_{a\in A}u(a).
		\end{align}
	\end{lemma}
	
	\begin{proof}
		By the definition of $W_{\gamma}$, we have
		\begin{align}
			W_{\gamma}	&= \frac{1}{t_2-t_1}\sum_{s={t_1+1}}^{t_2}\left(1-\g\varphi\left(a_s\middle\vert \vec a^{s-1}\right)\right)u(a_s)
			\notag \\
			&=
			\frac{1 }{t_2-t_1}\sum_{s={t_1+1}}^{t_2}	\sum_{a\in A} \1_{ a_s=a}u(a)
			- \frac{1 }{t_2-t_1}\sum_{s={t_1+1}}^{t_2}\g\varphi\left(a_s\middle\vert \vec a^{s-1}\right)u(a_s) \notag \\
			&=\sum_{a\in A} u(a)\left(\frac{1 }{t_2-t_1}\sum_{s={t_1+1}}^{t_2}\1_{ a_s=a}\right)
			- \frac{1 }{t_2-t_1}\sum_{s={t_1+1}}^{t_2}\g\varphi\left(a_s\middle\vert \vec a^{s-1}\right)u(a_s)	\notag
			\\
			&=	
			\sum_{a\in A}u(a)p(a)-\frac{1 }{t_2-t_1}\sum_{s={t_1+1}}^{t_2}\g\varphi\left(a_s\middle\vert \vec a^{s-1}\right)u(a_s). \label{eq: w}
		\end{align}
		Furthermore, for every $a\in A$ and $t_1+1\le s\le t_2$,
		\begin{align*}
			\LV\f\left(a\middle\vert\vec a^{s-1}\right)-\varphi\left(a\middle\vert \vec a^{t_1}\right)\RV
			&=
			\LV\frac{1}{s-1}\sum_{r=1}^{s-1}\1_{ a_r=a}- \frac{1}{t_1}\sum_{r=1}^{t_1}\1_{ a_r=a}\RV\\
			&=
			\LV \frac{1}{t_1} \frac{t_1}{s-1}\sum_{r=1}^{t_1}\1_{ a_r=a} - \frac{1}{ t_1} \sum_{r=1}^{t_1}\1_{ a_r=a}
			+\frac{1}{s-1}\sum_{r=t_1+1}^{s-1}\1_{ a_r=a}\RV\\
			&\leq
			\LV \frac{1}{t_1}\left(\frac{t_1}{s-1}-1\right) \sum_{r=1}^{t_1}\1_{ a_r=a} \RV
			+\frac{1}{s-1} \sum_{r=t_1+1}^{s-1}1\\
			&= \frac{s-1-t_1}{s-1}\varphi\left(a\middle\vert\vec a^{t_1}\right) + \frac{s-1-t_1}{s-1}\\
			&< 2\frac{t_2-t_1}{t_1},
		\end{align*}
		where in the last step we use that $t_1\leq s-1\leq t_2$ and $\varphi\left(a\middle\vert\vec a^{t_1}\right)\leq 1$.
		From (\ref{equ:approx_U2}) and (\ref{eq: w}) we now obtain
		\begin{align*}
			\big| W_{\gamma}-\widetilde{U}_{\gamma}\big| &= \gamma\LV\sum_{a\in A}p(a)\varphi\left(a\middle\vert\vec a^{t_1}\right)u(a) - \frac{1 }{t_2-t_1}\sum_{s={t_1+1}}^{t_2}\varphi\left(a_s\middle\vert\vec a^{s-1}\right)u(a_s)\RV\\
			&\leq \gamma\sum_{a\in A}\frac{u(a)}{t_2-t_1}\LV \varphi\left(a\middle\vert\vec a^{t_1}\right)\sum_{s=t_1+1}^{t_2}\mathbbm 1_{a_s=a} - \sum_{s=t_1+1}^{t_2}\mathbbm 1_{a_s=a}\varphi\left(a\middle\vert \vec a^{s-1}\right)\RV\\
			&\leq \gamma\sum_{a\in A}\frac{u(a)}{t_2-t_1}\sum_{s=t_1+1}^{t_2}\mathbbm 1_{a_s=a}\LV\varphi\left(a\middle\vert\vec a^{t_1}\right)-\varphi\left(a\middle\vert \vec a^{s-1}\right)\RV\\
                &<\gamma\sum_{a\in A}\frac{u(a)}{t_2-t_1}\sum_{s=t_1+1}^{t_2}\mathbbm 1_{a_s=a}2\frac{t_2-t_1}{t_1}\\
                &=2\frac{t_2-t_1}{t_1}\gamma\sum_{a\in A}u(a)p(a)\\
			&\leq 2\frac{t_2-t_1}{t_1}\gamma\sum_{a\in A}u(a).
		\end{align*}
		as required.
	\end{proof}

	\noindent With Lemma \ref{lem:w_U_approximation} we can now prove our first main result, namely that there is a stationary strategy such that the evolving history $\vec a$ satisfies $\underline V_{\gamma}=\underline V_{\gamma}\left(\vec a\right)=V^*_{\gamma}$.
	The idea of the proof is to suppose by contradiction that there is a non-stationary history $\vec a$ with $\underline V_{\gamma}\left(\vec a\right)\geq V^*_{\gamma} +c$ for some strictly positive constant $c$.
	This infinite history is split up into blocks that all satisfy the conditions of Lemma \ref{lem:w_U_approximation}, i.e., that are not too long.
	Denote by $\varphi^k(a)=\varphi\left(a\middle\vert \vec a^{t_k}\right)$ the frequency of $a$ at the beginning of the $k$-th block and for each block $k$ consider the number
	\begin{align*}
		x_k=\sum_{a\in A}\left(1-\varphi^k(a)\right)^2u(a).
	\end{align*}
	We show that these numbers would behave awkwardly if $\underline V_{\gamma}$ were bounded away from $V_{\gamma}^*$.
	Specifically, denote by $H_K$ the weighted average of $x_1,\ldots, x_K$, where the weight of $x_k$ equals the relative length of the $k$-th block within the first $K$-blocks.
	We then conclude that $\limsup_K H_K$ is bounded away from $\limsup_Kx_K$, which is impossible.

	\begin{theorem}\label{thm:liminf_stationary}
		It holds that $\underline V_{\gamma}=V^*_{\gamma}$.
	\end{theorem}
	
	\begin{proof}
		Assume, by contradiction, that $\underline V_{\gamma}>V^*_{\gamma}$.
		Then there is $\vec a\in A^{\infty}$ such that $\underline V\left(\vec a\right)= 4c+V_{\gamma}^*$ for some constant $c>0$.
		Thus,
		\begin{align}\label{equ:c}
			U^T_{\gamma}\left(\vec a\right)\geq V_{\gamma}^* +3c
		\end{align}
		for all sufficiently large $T$.
		Let $T_1$ be such that (\ref{equ:c}) holds for all $T\geq T_1$.
		
		Let $\a\in (0,1)$.
		We divide the set of periods into blocks.
		To that end let $t_0=0$, and for each integer $k\geq 1$ the let $t_k$ be defined by $ t_k=  \left\lceil(1+\a)^{k-1}T_1\right\rceil $, which is the smallest integer larger than or equal to $(1+\a)^{k-1}T_1$.
		The \emph{$k$-th block} starts at $t_{k-1}+1$ and ends at $ t_k$.
		For each block $k$, denote the average payoff, the frequency, and the approximation by $W^k_{\gamma}=W_\gamma\left(\vec a, t_{k-1},t_k\right)$, $p^k(a)=p\left(a;\vec a,t_{k-1},t_k\right)$, and $\widetilde{U}_{\gamma}^k=\widetilde U_{\gamma}\left(\vec a,t_{k-1},t_k\right)$, respectively, as in Equations (\ref{equ:W}), (\ref{equ:p}), and (\ref{equ:approx_U2}).
		In particular, $W^1_{\gamma}=U^{T_1}_\g$. 
		%
		%
		By construction, $\frac{t_{k+1}-t_{k}}{t_{k}}\leq \alpha + \frac{1}{\left(1+\alpha\right)^{k-1}T_1}$ for all $k\geq 1$.
		Thus, by Lemma \ref{lem:w_U_approximation},
		\begin{align*}
			\left\vert W^k_{\gamma}- \widetilde{U}^k_{\gamma}\right\vert\le 2\left(\alpha + \frac{1}{\left(1+\alpha\right)^{k-2}T_1}\right) \g \sum_{a\in A}u(a),
		\end{align*}
		for every $k\geq2$. Denote $\beta=\frac{\alpha}{1+\alpha}$ and also $d^k(a)=p^k(a)-\varphi^k(a)$ for each $k\in\mathbb{N}\setminus\{0\}$ and $a\in A$.
		Recall from (\ref{equ:c}) and the definition of $T_1$ that
		\begin{align}
			V_{\gamma}^*+ 3c\leq U^{t_K}_{\gamma}\left(\vec a\right)
		\end{align}
		for all $K\geq 1$.
		In particular,
        \small
		\begin{align}
			V_{\gamma}^*+ 3c &\leq U^{t_K}_{\gamma}\left(\vec a\right)\notag\\
			&= \frac{t_1}{t_K}W^1_{\gamma} + \sum_{k=2}^{K}\frac{t_{k}-t_{k-1}}{t_K}W^k_{\gamma}\notag\\
                &=\frac{T_1}{\left\lceil(1+
                \alpha)^{K-1}T_1\right\rceil}W^1_\gamma+\sum_{k=2}^{K}\frac{\left\lceil(1+\alpha)^{k-1}T_1\right\rceil-\left\lceil(1+\alpha)^{k-2}T_1\right\rceil}{\left\lceil(1+
                \alpha)^{K-1}T_1\right\rceil}W^k_\gamma\notag\\
                &\leq\frac{T_1}{(1+
                \alpha)^{K-1}T_1}W^1_\g+\sum_{k=2}^{K}\frac{\left\lceil(1+\alpha)^{k-2}T_1+\alpha(1+\alpha)^{k-2}T_1\right\rceil-\left\lceil(1+\alpha)^{k-2}T_1\right\rceil}{\left\lceil(1+
                \alpha)^{K-1}T_1\right\rceil}W^k_\gamma\notag\\
                &\leq \frac{1}{\left(1+\alpha\right)^{K-1}}W^1_{\gamma}+\sum_{k=2}^{K}\frac{\left\lceil(1+\alpha)^{k-2}T_1\right\rceil+\left\lceil\alpha(1+\alpha)^{k-2}T_1\right\rceil-\left\lceil(1+\alpha)^{k-2}T_1\right\rceil}{\left\lceil(1+
                \alpha)^{K-1}T_1\right\rceil}W^k_\gamma\notag\\
                &=\frac{1}{(1+
                \alpha)^{K-1}}W^1_\g+\sum_{k=2}^{K}\frac{\left\lceil\alpha(1+\alpha)^{k-2}T_1\right\rceil}{\left\lceil(1+
                \alpha)^{K-1}T_1\right\rceil}W^k_\gamma\notag\\
                &\leq\frac{1}{(1+
                \alpha)^{K-1}}W^1_\g+\sum_{k=2}^{K}\frac{\left\lceil\alpha(1+\alpha)^{k-2}\right\rceil\cdot T_1}{(1+
                \alpha)^{K-1}T_1}W^k_\gamma\notag\\
                &\leq\frac{1}{(1+
                \alpha)^{K-1}}W^1_\g+\sum_{k=2}^{K}\frac{\alpha(1+\alpha)^{k-2}+1}{(1+
                \alpha)^{K-1}}W^k_\gamma\notag\\
                &\leq\frac{1}{(1+
                \alpha)^{K-1}}W^1_\g+\sum_{k=2}^{K}\frac{\alpha(1+\alpha)^{k-2}+1}{(1+
                \alpha)^{K-1}}\left(\tilde U^k_{\gamma} + 2\left(\alpha + \frac{1}{\left(1+\alpha\right)^{k-2}T_1}\right)\gamma\sum_{a\in A} u(a)\right)\notag\\
			&\leq\frac{1}{(1+
                \alpha)^{K-1}}W^1_\g+\sum_{k=2}^K\frac{\left(1+\alpha\right)^{k-2}\alpha+1}{\left(1+\alpha\right)^{K-1}}
			\widetilde{U}^k_{\gamma} \notag\\
            & \qquad\qquad\qquad +2\g\sum_{a\in A}u(a)\sum_{k=2}^K\frac{(1+\alpha)^{k-2}\alpha+1}{(1+\alpha)^{K-1}}\left(\alpha+\frac{1}{(1+\alpha)^{k-2}}\right)\label{equ:V_gamma_3c},
	   \end{align}
        \normalsize
        for every $K\geq2$.
		\ignore{where the last inequality uses that
		\begin{align*}
			\sum_{k=1}^{K-1}\frac{\left(1+\alpha\right)^{k-1}\alpha}{\left(1+\alpha\right)^K}2\left(\alpha + \frac{1}{\left(1+\alpha\right)^kT_1}\right)\gamma\sum_{a\in A} u(a)
			&\leq \frac{ 2\left(\alpha + \frac{1}{T_1}\right)\gamma\sum_{a\in A} u(a)}{\left(1+\alpha\right)^K} \alpha\sum_{k=1}^{K-1}\left(1+\alpha\right)^{k-1}\\
			&\leq \frac{ 2\left(\alpha + \frac{2}{T_1}\right)\gamma\sum_{a\in A} u(a)}{\left(1+\alpha\right)^K}.
		\end{align*}
  }
		Let $\bar{\alpha}=c\left(4\gamma\sum_{a\in A} u(a)\right)^{-1}$. Then, for every $0<\a<\bar{\alpha}$ there is $K\left(\alpha\right)$ such that for all $K\geq K\left(\alpha\right)$ it holds that
		\begin{align*}
			 2\g\sum_{a\in A}u(a)\sum_{k=2}^K\frac{(1+\alpha)^{k-2}\alpha+1}{(1+\alpha)^{K-1}}\left(\alpha+\frac{1}{(1+\alpha)^{k-2}}\right)<c.
		\end{align*}
           This is so because
          $$
                \sum_{k=2}^K\frac{(1+\alpha)^{k-2}\alpha+1}{(1+\alpha)^{K-1}}\left(\alpha+\frac{1}{(1+\alpha)^{k-2}}\right)\longrightarrow[K\to\infty]{}\alpha.
           $$
		By (\ref{equ:V_gamma_3c})  we, hence, have
		\begin{align}\label{eq: c bound}
			V_{\gamma}^*+ 2c \le W^1_\g+\sum_{k=2}^K\frac{\left(1+\alpha\right)^{k-2}\alpha+1}{\left(1+\alpha\right)^{K-1}}
			\widetilde{U}^k_{\gamma}.
		\end{align}
		For every $x,y\in \R^A$, denote
		\begin{align}\label{equ:inner_product}
			\langle x,y\rangle &:= \sum_{a\in A}x_a y_a u(a)  & \rm{and} & &  \| x\|^2 &:= \langle x,x\rangle.
		\end{align}
		For each $k$ and $a\in A$ define $d^k(a)=p^k(a)-\varphi^k(a)$, and define $\beta=\frac{\alpha}{1+\alpha}$.
		Since $\varphi^k(\cdot)$ is a convex combination of $\varphi^{k-1}\left(\cdot\right)$ and $p^{k-1}\left(\cdot\right)$, with weights $\frac{t_{k-1}}{t_k}$ and $\frac{t_k-t_{k-1}}{t_k}$, we have
		\begin{align}
			\varphi^k(a) -\varphi^{k-1}(a) &= \frac{t_{k-1}}{t_k}	\varphi^{k-1}(a) + \frac{t_k-t_{k-1}}{t_k}p^{k-1}(a)-\varphi^{k-1}(a)\notag\\
			&= \frac{t_k-t_{k-1}}{t_k}\left(p^{k-1}(a)-\varphi^{k-1}(a)\right)\notag\\
			&\geq \frac{\left(1+\alpha\right)^k - \left(1+\alpha\right)^{k-1}-1}{\left(1+\alpha\right)^k}d^k(a)\notag\\
			&\geq \beta d^k(a) - \frac{1}{\left(1+\alpha\right)^k}\label{equ:recursive_phi_beta}.
		\end{align}
		
		\noindent Let $p^k= \left(p^k(a)\right)_a $, $ \varphi^{k}= \left(\varphi^k(a)\right)_a$ and $d^k= \left(d^k(a)\right)_a$, so that $p^k,\varphi^k,d^k\in \R^A$.
        Moreover, denote by $\mathbf 1\in\mathbb R^A$ the vector with 1 in each entry.
		By (\ref{equ:max_problem_stationary})
		\begin{align*}
			V^*_{\gamma} = \sup_{x\in\Delta(A)}\left\langle x,\mathbf 1-\gamma x\right\rangle \geq \left\langle\varphi^{t_k},  \mathbf 1-{\gamma}\varphi^{t_k}\right\rangle
		\end{align*}
		for all $k$.
		Since $\widetilde{U}_{\gamma}^k = \left\langle p^k,  \mathbf 1-{\gamma}\varphi^{k}\right\rangle$, we obtain from the definition of $d^k$ and Inequality (\ref{eq: c bound}) that for all  $0<\a<\bar{\alpha}$ and $K\geq K\left(\alpha\right)$
		\begin{align*}
			V_{\gamma}^* + 2c &\leq  \sum_{k=1}^K\frac{\left(1+\alpha\right)^{k-1}\alpha}{\left(1+\alpha\right)^K}\langle p^k,  \mathbf 1-{\gamma}\varphi^{k}\rangle\\
			&= \sum_{k=1}^K\frac{\left(1+\alpha\right)^{k-1}\alpha}{\left(1+\alpha\right)^K}
			\langle\varphi^{t_k},  \mathbf 1-{\gamma}\varphi^{k}\rangle  + 	\sum_{k=1}^K\frac{\left(1+\alpha\right)^{k-1}\alpha}{\left(1+\alpha\right)^K}
			\langle d^k ,  \mathbf 1-{\gamma}\varphi^{k}\rangle\\
			&\le  V_{\gamma}^* + \sum_{k=1}^K\frac{\left(1+\alpha\right)^{k-1}\alpha}{\left(1+\alpha\right)^K}
			\langle d^k ,  \mathbf 1-{\gamma}\varphi^{k}\rangle.
		\end{align*}
		Hence,
		\begin{align}
			2c  & \leq \sum_{k=1}^K\frac{\left(1+\alpha\right)^{k-1}\alpha}{\left(1+\alpha\right)^K}
			\langle d^k ,  \mathbf 1-{\gamma}\varphi^{k}\rangle = \sum_{k=1}^K\frac{\alpha}{\left(1+\alpha\right)^{K-k+1}}\langle d^k, \mathbf 1-{\gamma}\varphi^{k}\rangle\notag \\& = \frac{1}{1+\alpha}\sum_{k=1}^K\frac{\a}{\left(1+\alpha\right)^{K-k}}\langle d^k , \mathbf 1-{\gamma}\varphi^{k}\rangle. \label{equ:inequality_c}
		\end{align}
		For any $K$ and $\alpha$ define
		\begin{align*}
			H(K,\alpha) = \sum_{k=1}^K\frac{\a}{\left(1+\alpha\right)^{K-k+1}}\big\| \mathbf 1-{\gamma}\varphi^{k}\big\|^2 .
		\end{align*}
		Note that $\sum_{k=1}^K\frac{\a}{\left(1+\alpha\right)^{K-k+1}}\leq 1$, so that $H(K,\alpha)$ is bounded by some weighted average of $\left\| 1-{\gamma}\varphi^{k}\right\|^2 $ where $k=1,2,..., K$.
		Furthermore,
		\begin{align*}
			H(K,\alpha)& = \sum_{k=1}^{K-1}\frac{\a}{\left(1+\alpha\right)^{K-k+1}}\big\| \mathbf 1-{\gamma}\varphi^{k}\big\|^2 +\frac{\a}{\left(1+\alpha\right)}\big\| \mathbf 1-{\gamma}\varphi^{K}\big\|^2 \notag\\
			& = \left(1-\beta\right)H(K-1,\alpha)+\beta \left\| \mathbf 1-{\gamma}\varphi^{K}\right\|^2.
		\end{align*}
		Therefore,
		\begin{align}\label{eq: compare H}
			H(K,\alpha)-H(K-1,\alpha)=  -\beta H(K-1,\alpha)+\beta\left\| \mathbf 1-{\gamma}\varphi^{K}\right\|^2.
		\end{align}
		Define
		\begin{align*}
			\varepsilon_{K,\alpha} =  \frac{\a}{\left(1+\alpha\right)^K} \left\| \mathbf 1-{\gamma}\varphi^{1}\right\|^2.
		\end{align*}
		Then
        \small
		\begin{align*}
			H(K,\alpha)
            &= \varepsilon_{K,\alpha}+ \sum_{k=2}^K\frac{\a}{\left(1+\alpha\right)^{K-k+1}} \big\| \mathbf 1-{\gamma}\varphi^{k}\big\|^2 \notag \\
			&\leq 	\varepsilon_{K,\alpha}+ \sum_{k=2}^K\frac{\a}{\left(1+\alpha\right)^{K-k+1}}
			\big\|\mathbf 1- \g\varphi^{{k-1}} - \g\beta d^{k-1}\big\|^2 +\sum_{k=2}^K\frac{\a}{\left(1+\alpha\right)^{K-k+1}} \frac{1}{\left(1+\alpha\right)^k} \notag \\
			&= \varepsilon_{K,\alpha}+  \sum_{k=1}^{K-1}\frac{\a}{\left(1+\alpha\right)^{K-k}}
			\big\|\mathbf 1- \g\varphi^{{k}} - \g\beta d^{k}\big\|^2 +\frac{(K-2)\a}{\left(1+\alpha\right)^{K+1}}  \notag \\
			&=	\varepsilon_{K,\alpha}+\sum_{k=1}^{K-1}\frac{\a}{\left(1+\alpha\right)^{K-k}}\left(\big\|\mathbf 1- \g\varphi^{{k}}\big\|^2 - 2\g\beta \left\langle d^{k},  \mathbf 1-{\gamma}\varphi^{{k}}\right\rangle + \g^2\beta^2\big\|d^k\big\|^2\right) 
            +\frac{(K-2)\a}{\left(1+\alpha\right)^{K+1}}  \notag \\
			&= \varepsilon_{K,\alpha} +H(K-1,\a)  - 2\g\beta  \sum_{k=1}^{K-1}\frac{\a}{\left(1+\alpha\right)^{K-k}}
			\left\langle d^{k},  \mathbf 1-{\gamma}\varphi^{{k}}\right\rangle + \g^2\beta^2\sum_{k=1}^{K-1}\frac{\a}{\left(1+\alpha\right)^{K-k}}\big\|d^k\big\|^2  \notag \\
			& \qquad\qquad\qquad+\frac{(K-2)\a}{\left(1+\alpha\right)^{K+1}}. 
		\end{align*}
        \normalsize
		Thus,
		\begin{align*}
			H(K,\alpha)	- H(K-1,\a) &\leq \varepsilon_{K,\alpha}  - 2\beta \g \sum_{k=1}^{K-1}\frac{\a}{(1+\alpha)^{K-k}}
			\langle d^{k},  \mathbf 1-{\gamma}\varphi^{{k}}\rangle  \\
            & \qquad\qquad\qquad+ \beta^2\g^2\sum_{k=1}^{K-1}\frac{\a}{\left(1+\alpha\right)^{K-k}}\|d^k\|^2 +\frac{(K-2)\a}{\left(1+\alpha\right)^{K+1}}
		\end{align*}	
        and together with (\ref{equ:inequality_c}) and (\ref{eq: compare H}), we get
		\begin{align*}	
			\beta H(K-1,\alpha)-\beta \left\| \mathbf 1-{\gamma}\varphi^{t_K}\right\|^2
			&\geq -\varepsilon_{K,\alpha}  + 2\beta \g \sum_{k=1}^{K-1}\frac{\a}{\left(1+\alpha\right)^{K-k}}
			\left\langle d^{k},  \mathbf 1-{\gamma}\varphi^{{k}}\right\rangle \\
			& \qquad\qquad\qquad - \beta^2\g^2\sum_{k=1}^{K-1}\frac{\a}{\left(1+\alpha\right)^{K-k}}\big\|d^k\big\|^2 - \frac{(K-2)\a}{\left(1+\alpha\right)^{K+1}}\\
			& > -\varepsilon_{K,\alpha}  +4\beta\g(1+\a)c \\
			& \qquad\qquad\qquad - \beta^2\g^2\sum_{k=1}^{K-1}\frac{\a}{\left(1+\alpha\right)^{K-k}}\big\|d^k\big\|^2 - \frac{(K-2)\a}{\left(1+\alpha\right)^{K+1}},
		\end{align*}
		or equivalently,
        \small
		\begin{align*}
			H(K-1,\alpha) \! -\! \left\| \mathbf 1-{\gamma}\varphi^{K}\right\|^2 &> -\frac{\varepsilon_{K,\alpha}}{\beta}  +4\g(1+\a)c - \beta\g^2\sum_{k=1}^{K-1}\frac{\a}{\left(1+\alpha\right)^{K-k}}\big\|d^k\big\|^2 - \frac{(K-2)\a}{\beta\left(1+\alpha\right)^{K+1}}\\
			&= 4\g(1+\a)c -\frac{\alpha}{1+\alpha}\g^2\sum_{k=1}^{K-1}\frac{\a}{\left(1+\alpha\right)^{K-k}}\big\|d^k\big\|^2 - \left(\frac{\varepsilon_{K,\alpha}}{\beta}  +  \frac{(K-2)}{\left(1+\alpha\right)^{K}}\right).
		\end{align*}
		\normalsize
        Since $\left\|d^k\right\|^2$ are all uniformly bounded, the sum on the right-hand side is bounded.
		Thus,
		\begin{align*}
			\frac{\alpha}{1+\alpha}\g^2\sum_{k=1}^{K-1}\frac{\a}{\left(1+\alpha\right)^{K-k}}\big\|d^k\big\|^2
			&\leq  \frac{\alpha}{1+\alpha}\g^2\sup_{k\in\mathbb N}\left\|d^k\right\|^2 \sum_{k=0}^{\infty}\frac{\alpha}{\left(1+\alpha\right)^k}\\
			&= \frac{\alpha}{1+\alpha}\g^2\sup_{k\in\mathbb N}\big\|d^k\big\|^2\\
			&< c\g
		\end{align*}
		for all $K$ and all sufficiently small $\alpha>0$.
		Moreover, there are $\alpha^*$ and $K^*\geq K\left(\alpha^*\right)$ such that for all $K\geq K^*$
		\begin{align*}
			\frac{\varepsilon_{K,\alpha^*}}{\beta^*} + \frac{(K-2)}{\left(1+\alpha^*\right)^{K}}
            = \frac{1}{\left(1+\alpha^*\right)^{K-1}} \left\| \mathbf 1-{\gamma}\varphi^{1}\right\|^2 + \frac{(K-2)}{\left(1+\alpha^*\right)^{K}}< c\g,
		\end{align*}
		where $\beta^*=\frac{\alpha^*}{1+\alpha^*}$.
		This implies that
		\begin{align*}
			H\left(K-1,\alpha\right)-
			\left\| \mathbf 1-{\gamma}\varphi^{K}\right\|^2 >  4\gamma\left(1+\alpha\right) c-\gamma c -\gamma c >2\gamma c
		\end{align*}
		for all $\alpha\leq\alpha^*$ and all $K\geq K^*$.
		Consequently, due to (\ref{eq: compare H}), for $\alpha\leq\alpha^*$ we have,
		\begin{align*}
			0 &=\limsup_{K\rightarrow\infty} (  H\left(K,\alpha\right) - H\left(K-1,\alpha\right))\\
            &= \limsup_{K\rightarrow\infty} \beta( \left\| \mathbf 1-{\gamma}\varphi^{K}\right\|^2-H(K-1,\alpha))  \leq -2\beta \g c<0.
		\end{align*}
		We reached a contradiction.
	\end{proof}

	\subsection{Increasing fatigue}\label{subsec:fatique}
	
	As we have shown that $\underline V_{\gamma}$ can be achieved using a stationary strategy, we shall now have a closer look into how the frequencies of optimal histories change as $\gamma$ varies.
	Intuitively, a larger $\gamma$ forces good actions to be used less often so that their stage payoff does not wear down too much.
	The following lemma makes this formal.
	As $\gamma$ increases, the aggregated weight on the top actions is decreasing.
	
	\begin{lemma}\label{pro:fatique}
		Let $A=\left\{a_1,\ldots, a_m\right\}$ with $u\left(a_1\right)\geq u\left(a_2\right)\geq \cdots \geq u\left(a_m\right)$.
		For each $\gamma\in\left(0,1\right]$ let $x_{\gamma}\in\Delta(A)$ be the (unique) solution to the maximization problem in (\ref{equ:max_problem_stationary}).
		Then
		\begin{align}\label{equ:gamma_fosd}
			\sum_{i=1}^k\frac{d}{d\gamma}x_{\gamma}\left(a_i\right)\leq 0
		\end{align}
		for all $k=1,\ldots,m$.
	\end{lemma}
	
	\begin{proof}
		First, observe that if $x_{\gamma}\left(a_i\right)=0$ for some $i$, $x_{\gamma}\left(a_j\right)=0$ for all $j\geq i$.
		Indeed, if $u\left(a_i\right)>u\left(a_{j}\right)$, this is immediately clear.
		If $u\left(a_i\right)=u\left(a_{j}\right)$, then let $y_{\gamma}\left(a_i\right)=x_{\gamma}\left(a_j\right)$, $y_{\gamma}\left(a_j\right)=x_{\gamma}\left(a_i\right)$ and $y_{\gamma}\left(a\right)=x_{\gamma}\left(a\right)$ for all $a\neq a_i,a_j$.
		Then $y_{\gamma}$ is a solution of (\ref{equ:max_problem_stationary}), contradicting uniqueness.
		This means that $\sum_{i=1}^{k^*}\frac{d}{d\gamma}x^{\gamma}\left(a_i\right)\leq 0$, with equality if $\lambda_{k^*+1}>0$, and $\sum_{i=1}^{k}\frac{d}{d\gamma}x^{\gamma}\left(a_i\right)= 0$ for all $k\geq k^*+1$.
		
		It remains to prove the claim for $k<k^*$.
		Let $A^*=\left\{a\in A:\: x^{\gamma}(a)>0\right\}=\left\{a_1,\ldots, a_{k^*}\right\}$.
		The Lagrangian of maximization problem (\ref{equ:max_problem_stationary}) is
		\begin{align*}
			\sum_{i=1}^m x\left(a_i\right)\left(1-x\left(a_i\right)\right)u\left(a_i\right) + \lambda_ix\left(a_i\right) - \mu\left(\sum_{i=1}^mx\left(a_i\right)-1\right),
		\end{align*}
		with first-order conditions
		\begin{align*}
			u\left(a_i\right)\left(1-2\gamma x\left(a_i\right)\right)+\lambda_i-\mu=0.
		\end{align*}
		for $i=1,\ldots,m$.
		For all $i$ we either have $x\left(a_i\right)=0$ or $\lambda_i=0$; in the latter case
		\begin{align*}
			u\left(a_i\right)\left(1-2\gamma x\left(a_i\right)\right)=\mu.
		\end{align*}
		Summing over all $a_i\in A^*$, solving for $\mu$ and substituting in we find that
		\begin{align*}
			u\left(a\right)\left(1-2\gamma x\left(a\right)\right) = \frac{1}{k^*}\sum_{i=1}^{k^*}u\left(a_i\right)\left(1-2\gamma x\left(a_i\right)\right) = \frac{1}{k^*}\sum_{i=1}^{k^*}u\left(a_i\right) - 2\gamma \sum_{i=1}^{k^*}x\left(a_i\right)u\left(a_i\right).
		\end{align*}
		So, $x_{\gamma}$ satisfies for all $k=1,\ldots, k^*$,
		\begin{align*}
			x_{\gamma}\left(a_k\right) = \frac{1}{2\gamma u\left(a_k\right)}\left(u\left(a_k\right) - \frac{1}{k^*}\sum_{i=1}^{k^*}u\left(a_i\right)\right) + \frac{1}{u\left(a_k\right)}\sum_{i=1}^{k^*}x^{\gamma}\left(a_i\right)u\left(a_i\right).
		\end{align*}
		Taking the derivative of both sides with respect to $\gamma$ gives
		\begin{align}\label{equ:gamma_derivative}
			\frac{d}{d\gamma}x^{\gamma}\left(a_k\right) = \frac{1}{2\gamma^2}\left( \frac{1}{k^*}\sum_{i=1}^{k^*}u\left(a_i\right) - u\left(a_k\right)\right) + \frac{1}{u\left(a_k\right)}\sum_{i=1}^{k^*}\frac{d}{d\gamma}x^{\gamma}\left(a_i\right)u\left(a_i\right).
		\end{align}
		Suppose first that $\sum_{i=1}^{k^*}\frac{d}{d\gamma}x^{\gamma}\left(a_i\right)u\left(a_i\right)>0$.
		Then the right-hand side of (\ref{equ:gamma_derivative}) is increasing in $k$, as $u\left(a_k\right)$ is decreasing in $k$.
		Thus, in particular,
		\begin{align}\label{equ:gamma_ordered_derivatives}
			\frac{d}{d\gamma}x^{\gamma}\left(a_1\right) \leq \frac{d}{d\gamma}x^{\gamma}\left(a_2\right) \leq \cdots \leq \frac{d}{d\gamma}x^{\gamma}\left(a_{k^*}\right).
		\end{align}
		Suppose that (\ref{equ:gamma_fosd}) does not hold.
		Then there is $k<k^*$ such that $\sum_{i=1}^{k}\frac{d}{d\gamma}x^{\gamma}\left(a_i\right)> 0$.
		Thus, by (\ref{equ:gamma_ordered_derivatives}), we must have $\sum_{i=1}^{k^*}\frac{d}{d\gamma}x^{\gamma}\left(a_i\right)\geq \sum_{i=1}^{k}\frac{d}{d\gamma}x^{\gamma}\left(a_i\right)> 0$, which is impossible.
		Suppose next that $\sum_{i=1}^{k^*}\frac{d}{d\gamma}x^{\gamma}\left(a_i\right)u\left(a_i\right)\leq 0$.
		Then, for all $\ell\leq k^*$,
		\begin{align*}
			\sum_{k=1}^\ell \frac{d}{d\gamma}x^{\gamma}\left(a_k\right)
			&= \sum_{k=1}^\ell \frac{1}{2\gamma^2}\left( \frac{1}{k^*}\sum_{i=1}^{k^*}u\left(a_i\right) - u\left(a_k\right)\right) + \sum_{k=1}^\ell \frac{1}{u\left(a_k\right)}\sum_{i=1}^{k^*}\frac{d}{d\gamma}x^{\gamma}\left(a_i\right)u\left(a_i\right)\\
			&\leq \frac{1}{2\gamma^2} \left(\frac{\ell}{k^*}\sum_{i=1}^{k^*}u\left(a_i\right) - \sum_{k=1}^\ell u\left(a_k\right)\right)\\
			&\leq 0,
		\end{align*}
		where the last inequality holds because $\frac{1}{k^*}\sum_{i=1}^{k^*}u\left(a_i\right) \leq \frac{1}{\ell}\sum_{k=1}^\ell u\left(a_k\right)$ for all $\ell\leq k^*$.
	\end{proof}
	
	\noindent As $\sum_{a\in A}x_{\gamma}(a)=1$ for all $\gamma\in\left(0,1\right]$, an immediate consequence of Lemma \ref{pro:fatique} is that the aggregated weight on the poor actions is increasing as $\gamma$ increases.
	So, for comparing any two value $\gamma,\gamma'$ we obtain the following corollary.

	\begin{corollary}
		Let $A=\left\{a_1,\ldots, a_m\right\}$ with $u\left(a_1\right)\geq u\left(a_2\right)\geq \cdots \geq u\left(a_m\right)$, let $0<\gamma<\gamma'\leq 1$, and let $\vec a,\vec b\in A^{\infty}$ be two optimal stationary histories with respect to $\gamma$ and $\gamma'$, respectively.
		Then $\varphi\left(\cdot\middle\vert\vec a\right)$ first order stochastically dominates $\varphi\big(\cdot\big\vert\vec b\big)$ .
	\end{corollary}

	\section{Stationary strategies do not achieve $\overline V_{\gamma}$}\label{sec:limsup}
	
	In Examples \ref{exa:limsup} and \ref{exa:opt_stat} we have seen a set of actions for which $\overline V_{\gamma}>V^*_{\gamma}$.
	This relation is quite robust, as we will show in this section: whenever $A$ contains at least two actions with different basic payoffs, it holds true.
	The rough idea of the proof is to construct a sequence $\left(v^T_{\gamma}\right)_{T\in\mathbb N}$ with $\lim_{T\rightarrow\infty}v^T_{\gamma} = \overline V_{\gamma}$, and then show that $\lim_{T\rightarrow\infty}v^T_{\gamma} $ is bounded away from $V_{\gamma}^*$.
	
	\subsection{Approximating $\overline V_{\gamma}$}
	
	For every $T\geq 1$ define
	\begin{align*}
		v^T_{\gamma} = \max_{\vec a\in A^{\infty}}U^T_{\gamma}\left(\vec a\right).
	\end{align*}
	
	\noindent That is, $v_{\gamma}^T$ denotes the maximal average payoff that can be obtained from a history of length $T$.
	We show that the sequence $\left(v^T_{\gamma}\right)_{T\in\mathbb N}$ converges.
	The idea of the proof is to show that for a history $\vec a$ of length $T$ and any $S>T$ we can find a history $\vec b$ of length $S$ that has an average payoff at $S$ that is close to the one of $\vec a$ at $T$.
	The construction of $\vec b$ relies on the division of $\vec a$ into blocks such that the length of any block is a fraction $\alpha$ of the previous history, similar to the construction in the proof of Theorem \ref{thm:liminf_stationary}.
	These blocks are then ``stretched'' by some factor $\delta>1$ such that $S=\delta T$, and the within-block frequencies and average payoffs of the $k$-block of $\vec a$ and $\vec b$ are close.

	\begin{proposition}\label{pro:v_T_converges}
		The sequence $\left(v^T_{\gamma}\right)_{T\in\mathbb N}$ converges.
	\end{proposition}
	
	\begin{proof}
		Clearly, the sequence is bounded, so that $\limsup_{T\rightarrow\infty}v_{\gamma}^{T}$ and $\liminf_{T\rightarrow\infty}v_{\gamma}^{T}$ exist.
		We show that for each $\varepsilon>0$ there is $T^*\in\mathbb N$ such that if $v_{T^*}\geq\limsup_{T\rightarrow\infty}v_{\gamma}^{T}-\varepsilon$, then $v_S\geq \limsup_{T\rightarrow\infty}v_{\gamma}^{T}-2\varepsilon$ for all $S\geq T^*$.
		This implies that for every $\varepsilon>0$ it holds that $\limsup_{T\rightarrow\infty}v_{T}-\liminf_{T\rightarrow\infty}v_{T}<2\varepsilon$, so that $\limsup_{T\rightarrow\infty}v_{\gamma}^{T}=\liminf_{T\rightarrow\infty}v_{\gamma}^{T}=\lim_{T\rightarrow\infty}v^T_{\gamma}$.
		
		So, let $\varepsilon>0$ be sufficiently small.
		Let
		\begin{align*}
			t_1 &\geq \max\left\{\frac{1}{\varepsilon^3}\left(\LV A\RV + 8\gamma\sum_{a\in A}u(a)\right) , \frac{1+2\varepsilon^2}{\varepsilon^3}16\gamma\sum_{a\in A}u(a)\right\},\\
			\alpha &=\frac{\varepsilon}{16\gamma\sum_au(a)}-\frac{2}{t_1},
		\end{align*}
		and let $T^* \geq\frac{4t_1\sum_au(a)}{\varepsilon}$ be such that $v_{T^*}\geq\limsup_{T\rightarrow\infty}v_{T}-\varepsilon$.
		Observe that for sufficiently small  $\varepsilon$ we have
		\begin{align}
			1 \geq \alpha = \frac{\varepsilon}{16\gamma\sum_au(a)} - \frac{2}{t_1} \geq \frac{1+2{\varepsilon^2}}{\varepsilon^2t_1}- \frac{2}{t_1}=  \frac{1}{\varepsilon^2 t_1}>\frac{1}{t_1}.
		\end{align}
		For $k\geq 2$, let $r^k$ be the smallest integer such that $r^k\geq\left(1+\alpha\right)^{k-1}t_1$, and let $K$ be the smallest integer with $r^K\geq T^*$.
		Let $t^0=0$, for $k=2,\ldots,K-1$ let $t_k=r^k$, and let $t_K=T^*>t_{K-1}$.
		Let $\vec a\in A^{\infty}$ be such that $U^{T^*}_{\gamma}\left(\vec a\right)=v^{T^*}_{\gamma}$.
		As in the proof of Theorem \ref{thm:liminf_stationary}, let the $k$-th block of $\vec a$ be the finite sequence $\left(\vec a_{t_k+1},\ldots,\vec a_{t_{k+1}}\right)$.
		For $k=0,\ldots, K-1$, denote the average payoff, the frequency, and the approximation by $W^k_{\gamma}=W_\gamma\left(\vec a, t_k,t_{k+1}\right)$, $p^k(a)=p\left(a;\vec a,t_k,t_{k+1}\right)$, and $\widetilde{U}_{\gamma}^k=\widetilde U_{\gamma}\left(\vec a,t_k,t_{k+1}\right)$, respectively, as in Equations (\ref{equ:W}), (\ref{equ:p}), and (\ref{equ:approx_U2}).
		In particular, $W^0_{\gamma}=U^{t_1}$ and $p^0=\varphi\left(.\middle\vert \vec a^{t_1}\right)$.
		By construction, $\frac{t_{k+1}-t_k}{t_k}\leq\alpha+\frac{1}{t_1}$ for all $k=1,\ldots, K$.
		Thus, by Lemma \ref{lem:w_U_approximation} and the definition of $\alpha$,
		\begin{align}\label{equ:proof_pro:v^T_conv:Wk}
			\big\vert W^k- \widetilde{U}^k_{\gamma}\big\vert\le 2\left(\alpha+\frac{1}{t_1}\right) \g \sum_{a\in A}u(a) \leq 2\frac{\varepsilon}{16\gamma\sum_au(a)} \g \sum_{a\in A}u(a) = \frac{1}{8}\varepsilon
		\end{align}
		for $k=1,\ldots, K$.
		
		Let $S\geq T^*$ and define $\delta=\frac{S}{T^*}$.
		For each $k\geq 1$, let $s^k$ be the largest integer smaller than or equal to $\delta t_k$.
		Let $\vec b\in A^{\infty}$ be such that $\vec b^{t_1}=\vec a^{t_1}$ and
		\begin{align}\label{equ:construction_of_b1}
			\vec b_s =
			\begin{cases}
				\argmin_{a\in A}\varphi\left(a\mid b^{s-1}\right) - \varphi\left(a\mid a^{t_1}\right),  & \text{if } s=t_1+1,\ldots, s^1,\\
				\argmin_{a\in A}\frac{1}{s-s^{k}}\sum_{s'=s^{k}+1}^{s}\1_{\vec b_{s'}=a} - p^k(a), &  \text{if } s^{k}+1\leq s\leq s^{k+1}, \text{ where } k\geq 1.
			\end{cases}		
		\end{align}
		That is, the individual blocks of history $\vec b$ are longer than those of $\vec a$, stretched by the factor $\delta$, and in the $k$-th block of $\vec b$ actions are chosen to minimize the difference between the frequencies in the $k$-th block of $\vec a$ and $\vec b$.
		For $k=0,\ldots, K-1$ denote the average payoff, the frequency, and the approximation in $\vec b$ by $Y^k_{\gamma}=W_\gamma\big(\vec b, t_k,t_{k+1}\big)$, $q^k(a)=p\big(a;\vec b,t_k,t_{k+1}\big)$, and $\widetilde{V}_{\gamma}^k=\widetilde U_{\gamma}\big(\vec b,t_k,t_{k+1}\big)$, respectively, as in Equations (\ref{equ:W}), (\ref{equ:p}), and (\ref{equ:approx_U2}).
		As $\frac{s^{k+1}-s^k}{s^k}\leq\alpha+\frac{2}{t_1}$ for all $k\geq 1$, Lemma \ref{lem:w_U_approximation} together with the definition of $\alpha$ give
		\begin{align}\label{equ:proof_pro:v^T_conv:Yk}
			\left\vert Y^k- \widetilde{V}^k_{\gamma}\right\vert\le 2\left(\alpha+\frac{2}{t_1}\right)  \g \sum_{a\in A}u(a) = 2\frac{\varepsilon}{16\gamma\sum_au(a)} \g \sum_{a\in A}u(a) = \frac{1}{8}\varepsilon.
		\end{align}
		By construction, $\varphi\left(.\middle\vert\vec a^{t_1}\right) = \varphi\big(.\big\vert\vec b^{t_1}\big)$.
		By (\ref{equ:construction_of_b1}), for all $t_1+1\leq s\leq s^1$, action $a$ is only chosen if $\varphi\big(a\big\vert \vec b^{s-1}\big) \leq \varphi\left(a\middle\vert \vec a^{t_1}\right)$.
		Thus, $\varphi\big(a\big\vert \vec b^{s}\big)\leq \varphi\left(a\middle\vert \vec a^{t_1}\right) +\frac{1}{s}$ for all $s\leq s^1$.
		Since $\sum_{a\in A}\varphi\big(a\big\vert\vec b^{s}\big)=1$, this implies that $\varphi\big(a\big\vert\vec b^{s}\big)\geq \varphi\left(a\middle\vert\vec a^{t_1}\right) -\frac{\LV A\RV -1}{s}$ for all $s\leq s^1$.
		Thus, for sufficiently small $\varepsilon>0$
		\begin{align}\label{equ:frec_a_b_1}
			\big| \varphi\left(a\middle\vert\vec a^{t_1}\right) - \varphi\big(a\big\vert\vec b^{s}\big)\big| \leq \frac{\LV A\RV -1}{s} \leq \frac{\LV A\RV -1}{t_1} \leq \varepsilon^3\frac{\LV A\RV -1}{\LV A\RV + 8\gamma\sum_au(a)}\leq \varepsilon^2
		\end{align}
		for all $s\leq s^1$.	
		Let $k\geq 1$ and $s^k+1\leq s\leq s^{k+1}$.
		By (\ref{equ:construction_of_b1}), action $a$ is only being played at $s$ if $\frac{1}{s-s^{k}}\sum_{s'=s^{k}+1}^{s}\1_{\vec b_{s'}=a} \leq p^k(a)$.
		Thus, $p^k(a)\leq \frac{1}{s-s^{k}} + \frac{1}{s-s^{k}} \sum_{s'=s^{k}+1}^{s}\1_{\vec b_{s'}=a}$.
		In particular, for $s=s^{k+1}$ it holds that
		\begin{align*}
			p^k(a)\leq\frac{1}{s^{k+1}-s^k} + \frac{1}{s^{k+1}-s^k} \sum_{s'=s^k+1}^{s^{k+1}}\1_{\vec b_{s'}=a}=\frac{1}{s^{k+1}-s^k} + q^k(a).
		\end{align*}
		Since $\sum_{a\in A}q^k(a)=\sum_{a\in A}p^k(a)=1$, this implies $p^k(a)\geq q^k(a)+\frac{\LV A\RV -1}{s^{k+1}-s^k}$, so that for sufficiently small $\varepsilon>0$
		\begin{align}
			\big| p^k(a) - q^k(a)\big| &\leq \frac{\LV A\RV -1}{s^{k+1}-s^k}\leq \frac{\LV A\RV -1}{t_{k+1}-t_k} \label{equ:p^k_1}\\
			&\leq \frac{\LV A\RV -1}{\alpha t_k-1}\leq \frac{\LV A\RV -1}{\alpha\left(1+\alpha\right)^{k}t_1-1} \notag\\
			&\leq \frac{\LV A\RV -1}{\frac{1}{\varepsilon^2} \left(1+\alpha\right)^k-k} \leq \varepsilon^2\frac{\LV A\RV -1}{\left(1+\alpha\right)^k-\varepsilon^2}\notag\\
			&\leq \frac{\varepsilon}{8\LV A\RV u(a)} \notag
		\end{align}
		for all $a\in A$ and all $k\geq 1$.
		In particular,
		\begin{align*}
			\sum_{a\in A}\big| p^k(a)-q^k(a)\big| u(a) &\leq
			\sum_{a\in A}\frac{\varepsilon}{8\LV A\RV u(a)}u(a)=\frac{1}{8}\varepsilon.
		\end{align*}
		Further, by using (\ref{equ:p^k_1}) we find for $k\geq 2$ that
        \small
		\begin{align*}
			\LV \varphi\left(a\middle\vert\vec a^{t_k}\right) - \varphi\big(a\big\vert\vec b^{s^k}\big)\RV &=
			\LV \frac{t_{k-1}}{t_k}\varphi\left(a\middle\vert\vec a^{t_{k-1}}\right) + \frac{t_k-t_{k-1}}{t_k}p^k(a) - \frac{s^{k-1}}{s^k} \varphi\big(a\big\vert\vec b^{s^{k-1}}\big) - \frac{s^{k}-s^{k-1}}{s^k}q^k(a)\RV\\
			&\leq \LV \frac{t_{k-1}}{t_k}\varphi\left(a\middle\vert\vec a^{t_{k-1}}\right) - \frac{\delta t_{k-1}-x_1}{\delta t_k-x_2}\varphi\big(a\big\vert\vec b^{s^{k-1}}\big)\RV\\
			& \qquad\qquad\qquad + \LV \frac{t_k-t_{k-1}}{t_k}p^k(a) - \frac{\delta t_k-x_2 -\delta t_{k-1}+x_1}{\delta t_k- x_2}q^k(a)\RV\\
			&\leq \frac{t_{k-1}}{t_k}\LV \varphi\left(a\middle\vert\vec a^{t_{k-1}}\right) - \varphi\big(a\big\vert\vec b^{s^{k-1}}\big)\RV + \frac{x_1}{s^k}\varphi\big(a\big\vert\vec b^{s^{k-1}}\big)\\
			& \qquad\qquad\qquad + \frac{t_k-t_{k-1}}{t_k}\LV p^k(a)-q^k(a)\RV + \frac{\LV x_1-x_2\RV}{s^k}q^k(a)\\
			&\leq \frac{t_{k-1}}{t_k}\LV \varphi\left(a\middle\vert\vec a^{t_{k-1}}\right) - \varphi\big(a\big\vert\vec b^{s^{k-1}}\big)\RV +
			\frac{2}{t_k} + \frac{t_k-t_{k-1}}{t_k}\frac{\LV A\RV -1}{t_{k+1}-t_k}\\
			&\leq \frac{t_{k-1}}{t_k}\LV \varphi\left(a\middle\vert\vec a^{t_{k-1}}\right) - \varphi\big(a\big\vert\vec b^{s^{k-1}}\big)\RV +
			\frac{\LV A\RV +1}{t_k} ,
		\end{align*}
        \normalsize
		where $x_1, x_2\leq 1$ are such that $s^{k-1}=\delta t_{k-1} - x_1$ and $s^{k}=\delta t_k - x_2$.
		We thus find inductively that, for sufficiently small $\varepsilon$,
		\begin{align*}
			\LV \varphi\left(a\middle\vert\vec a^{t_k}\right) - \varphi\big(a\big\vert\vec b^{s^k}\big)\RV
			&\leq \frac{t_1}{t_k}\LV \varphi\big(a\big\vert\vec a^{t^{1}}\big) - \varphi\big(a\big\vert\vec b^{s^{1}}\big)\RV + \left(\LV A\RV +1\right)\sum_{l=2}^k\frac{1}{t^l}\\
			&\leq \frac{t_1}{t_k} \varepsilon^2 + \frac{\LV A\RV +1}{t_1}\sum_{l=2}^k\left(\frac{1}{1+\alpha}\right)^l\\
			&\leq \frac{1}{(1+\alpha)^k} \varepsilon^2 + \frac{\LV A\RV +1}{\alpha t_1}\\
			&\leq \frac{1}{(1+\alpha)^k} \varepsilon^2 + \left(\LV A\RV +1\right)\varepsilon^2\\
			&\leq \left(\LV A\RV +2\right)\varepsilon^2\\
			&\leq \frac{\varepsilon}{8\LV A\RV u(a)}
		\end{align*}	
		for all $k\geq 1$, so that
		\begin{align*}
			\sum_{a\in A}\LV \varphi\left(a\middle\vert\vec a^{t_k}\right) - \varphi\big(a\big\vert\vec b^{s^k}\big)\RV \leq\frac{1}{8}\varepsilon
		\end{align*}
		for all $k\geq 0$.
		(The case $k=0$ follows from (\ref{equ:frec_a_b_1}).)
		Hence,
        \small
		\begin{align}
			\big|\widetilde U^k_{\gamma}-\widetilde V^k_{\gamma}\big|
			&= \LV \sum_{a\in A}p^k(a)\left(1-\varphi\left(a\middle\vert\vec a^{t_k}\right)\right)u(a) - \sum_{a\in A}q^k(a)\big(1-\varphi^k\big(a\big\vert\vec b^{s^k}\big)\big)u(a)\RV \notag \\
			& = \LV \sum_{a\in A}\big(p^k(a) - q^k(a)\big)\left(1-\varphi\left(a\middle\vert\vec a^{t_k}\right)\right)u(a) - \sum_{a\in A}q^k(a)\big(\varphi\big(a\big\vert\vec a^{t_k}\big)-\varphi\big(a\big\vert\vec b^{s^k}\big)\big)u(a)\RV \notag \\
			& \leq \sum_{a\in A}\LV p^k(a) - q^k(a)\RV u(a) - \sum_{a\in A}\LV \varphi\left(a\middle\vert\vec a^{t_k}\right)-\varphi\big(a\big\vert\vec b^{s^k}\big)\RV u(a) \notag \\
			&\leq \frac{1}{4}\varepsilon \label{equ:proof_pro:v^T_conv:UV}
		\end{align}
        \normalsize
		for all $k\geq 1$.
		From (\ref{equ:proof_pro:v^T_conv:Wk}), (\ref{equ:proof_pro:v^T_conv:Yk}), an (\ref{equ:proof_pro:v^T_conv:UV}) we find
		\begin{align*}
			\big| W^k_{\gamma} - Y^k_{\gamma}\big|
			\leq \big| W^k_{\gamma}  - \tilde U^k_{\gamma} \big| + \big| \tilde U^k_{\gamma}  -\tilde V^{k}_\gamma\big| + \big| \tilde V^{k}_\gamma - Y^k_{\gamma}\big|
			\leq \frac{1}{2}\varepsilon
		\end{align*}
		for all $k\geq 1$.
		Thus, recalling that $t_K=T^*$ and $s^K= S$, we have
        \small
		\begin{align*}
			\LV U^T_{\gamma}\left(\vec b\right) - U^S_{\gamma}\left(\vec a\right)\RV
			&\leq \sum_{k=0}^{K-1}\LV \frac{t_{k+1}-t_k}{T}W^k_{\gamma} - \frac{s^{k+1}-s^{k}}{S}Y^{k}_{\gamma}\RV 
			\\
			&\leq \sum_{k=0}^{K-1}\LV \frac{\delta t_{k+1}- \delta t_k}{\delta T}W^k_{\gamma} - \frac{s^{k+1}-s^{k}}{\delta T}Y^{k}_{\gamma}\RV 
			\\
			&\leq \sum_{k=0}^{K-1} \left(\LV \frac{\delta t_{k+1}- \delta t_k}{\delta T}\left(W^k_{\gamma} - Y^{k}_{\gamma}\right) \RV + \LV\frac{1}{\delta T}Y^{k}_{\gamma}\RV\right) 
			\\
			&\leq \sum_{k=0}^{K-1}  \frac{\left(1+\alpha\right)^{k+1}t_1-\left(1+\alpha\right)^{k}t_1+1}{\left(1+\alpha\right)^{K-1}t_1}\LV W^k_{\gamma} - Y^{k}_{\gamma} \RV + \frac{K}{\delta \left(1+\alpha\right)^{K-1}t_1}\sum_{a\in A}u(a) \\
			&\leq \frac{\varepsilon}{2}\sum_{k=0}^{K-1}  \left(\alpha\left(1+\alpha\right)^{k-K+1} +\frac{1}{\left(1+\alpha\right)^{K-1}t_1}\right) + \frac{K}{ \left(1+\alpha\right)^{K-1}t_1}\sum_{a\in A}u(a) \\
			&\leq \frac{\varepsilon}{2}\left(1+\frac{K}{\left(1+\alpha\right)^Kt_1}\right) + \frac{K}{\left(1+\alpha\right)^{K-1}t_1}\sum_{a\in A}u(a).
		\end{align*}
        \normalsize
		Using that
		\begin{align*}
			\frac{K}{\left(1+\alpha\right)^Kt_1}\leq \frac{K}{\left(1+\alpha\right)^{K-1}t_1} \leq \frac{K}{\left(1+\alpha (K-1)\right)t_1} \leq \frac{K}{\left(\alpha+\alpha K -\alpha\right)t_1} = \frac{1}{\alpha t_1} \leq \varepsilon^2
		\end{align*}	
		we conclude that, for sufficiently small $\varepsilon$,
		\begin{align*}
			\LV U^T_{\gamma}\left(\vec b\right) - U^S_{\gamma}\left(\vec a\right)\RV
			\leq \frac{\varepsilon}{2}\left(1+\varepsilon^2\right) + \varepsilon^2\sum_{a\in A}u(a)
			\leq \varepsilon.
		\end{align*}
		Thus,
		\begin{align*}
			v^S_{\gamma}\geq U^S_{\gamma}\left(\vec b\right) \geq U^T_{\gamma}\left(\vec a\right) - \varepsilon \geq \limsup_{T'}v_{T'} - 2\varepsilon
		\end{align*}
		as required.
	\end{proof}

	\noindent After we have established that $\lim_{T\rightarrow\infty}v^T_{\gamma}$ is well-defined, we shall now show that the sequence converges to $\overline{V}_{\gamma}$.

	\begin{proposition}\label{pro:v_T_limit}
		It holds that $\lim_{T\rightarrow\infty} v^T_{\gamma} = \overline{V}_{\gamma}$.
	\end{proposition}
	
	\begin{proof}
		Let $\varepsilon >0$ and let $\vec a$ be such that $\overline V_{\gamma}\left(\vec a\right)\geq\overline V_{\gamma}-\varepsilon$.
		Then there is a sequence $\left(T_k\right)_{k\in\mathbb N}$ such that
		\begin{align*}
			v^{T_k}_{\gamma} \geq U^{T_k}_{\gamma}\left(\vec a\right) \geq \overline V_{\gamma}\left(\vec a\right) -\varepsilon \geq \overline V_{\gamma}-2\varepsilon
		\end{align*}
		for all $k\in\mathbb N$.
		Thus, $\lim_{T\rightarrow\infty}v^T_{\gamma} = \lim_{k\rightarrow\infty}v^{T_k}_{\gamma}\geq \overline V_{\gamma}-2\varepsilon$.
		As $\varepsilon>0$ was arbitrary, we have $\lim_{T\rightarrow\infty} v^T_{\gamma} \geq \overline{V}_{\gamma}$.
		
		Assume that there is $c>0$ such that $\lim_{T\rightarrow\infty} v^T_{\gamma} \geq \overline{V}_{\gamma} + 4c$.
		Let $T^0$ be such that $v^{T'}_{\gamma}\geq \lim_{T\rightarrow\infty} v^T_{\gamma} - c$ for all $T'\geq T^0$.
		There is $T_1\geq T^0$ such that
		\begin{align}\label{equ:pro:v_T_limit:1}
			v_{\gamma}^{T'}\geq \lim_{T\rightarrow\infty} v^T_{\gamma} -c \geq \overline{V}_{\gamma} + 3c = \sup_{\vec a\in A^{\infty}}\limsup_T U^{T}_{\gamma}\left(\vec a\right) + 3c
		\end{align}
		for all $T'\geq T_1$.
		For each $\vec a$ there is $T_2\left(\vec a\right)\geq T_1$ such that $\limsup_T U^{T}_{\gamma}\left(\vec a\right)\geq U^{T'}\left(\vec a\right)-c$ for all $T'\geq T_2\left(\vec a\right)$.
		In particular,
		\begin{align}\label{equ:pro:v_T_limit:2}
			\sup_{\vec a\in A^{\infty}}\limsup_T U^{T}_{\gamma}\left(\vec a\right) + 3c \geq \sup_{\vec a\in A^{\infty}} U^{T_2\left(\vec a\right)}_{\gamma}\left(\vec a\right) +2c \geq  \lim_{T\rightarrow\infty} v^T_{\gamma} +c,
		\end{align}
		where the last inequality holds since $T_2\left(\vec\alpha\right)\geq T_1\geq T_0$ for all $\vec a\in A^{\infty}$.
		From (\ref{equ:pro:v_T_limit:1}) and (\ref{equ:pro:v_T_limit:2}) we obtain $v_{\gamma}^{T'}\geq \lim_{T\rightarrow\infty} v^T_{\gamma} +c$ for all $T'\geq T_1$,
        which is impossible as $\left(v^T_{\gamma} \right)$ converges by Proposition \ref{pro:v_T_converges}.
	\end{proof}

	\subsection{Establishing suboptimality of $V^*_{\gamma}$}
	
	The next main result shows that $\overline V_{\gamma}$ cannot be achieved by any stationary strategy, or, to be more precise, that the optimal stationary strategy achieves an average payoff that is strictly less than $\overline V_{\gamma}$.
	In the proof we start from a stationary history $\vec a$ and then iteratively construct a sequence of histories by only switching two actions in each step.
	We show that the effect of these switches is significant, i.e., that for the final history $\vec b$ it holds that $\overline V_{\gamma}\big(\vec b\big)>V^*_{\gamma}\left(\vec a\right)+\eta$ for some constant $\eta>0$.
	We start with the following lemma that captures the effect of such pairwise switches.
	
	\begin{lemma}\label{lem:swap_general}
		Let $a,b\in A$ and let $\vec{a},\vec{b}\in A^{\infty}$ be such that there are $s>t$ with $\vec a_t=\vec b_s=a$, $\vec a_s=\vec b_t=b$, $\vec a_{t'}=\vec b_{t'}$ for all $t'\neq t,s$, and $\vec a_{t'}\neq a$ for all $t'=t,\ldots, s-1$.
		Then
		\begin{align*}
			U^T_{\gamma}(\vec{b}) -  U^T_{\gamma}(\vec{a}) \geq \frac{\gamma\left(s-t\right)}{\left(s-1\right)T}\left(\f\left(a\middle\vert \vec a^{t-1}\right) u(a) - \f\left(b\middle\vert \vec a^{t-1}\right) u(b)\right)
		\end{align*}
		for all $T\geq s$.
	\end{lemma}
	
	\begin{proof}
		By the conditions on $\vec a$ and $\vec b$,
		\begin{align*}
			\sum_{r=1}^{s-1}\mathbbm 1_{\vec a_r=b} \geq \sum_{r=1}^{t-1}\mathbbm 1_{\vec b_r=b} = \sum_{r=1}^{t-1}\mathbbm 1_{\vec a_r=b}
		\end{align*}
		and
		\begin{align*}
			\sum_{r=1}^{s-1}\mathbbm 1_{\vec b_r=a} =  \sum_{r=1}^{t-1}\mathbbm 1_{\vec a_r=a}.
		\end{align*}
		Thus,
		\begin{align*}
			T\left(U^T_{\gamma}(\vec{b}) -  U^T_{\gamma}(\vec{a})\right) &= \big(1-\gamma\f\big(b\big\vert\vec b^{t-1}\big)\big)u(b) + \big(1-\gamma\f\big(a\big\vert\vec b^{s-1}\big)\big)u(a)\\
			& \qquad\qquad - \left(1-\gamma\f\left(a\middle\vert\vec a^{t-1}\right)\right)u(a) - \left(1-\gamma\f\left(b\middle\vert\vec a^{s-1}\right)\right)u(b)\\
			&= \gamma\left(\big(\f\left(b\middle\vert\vec a^{s-1}\right) - \f\big(b\big\vert\vec b^{t-1}\big)\big) u(b) + \big(\f\left(a\middle\vert\vec a^{t-1}\right) - \f\big(a\big\vert\vec b^{s-1}\big) \big) u(a)\right)\\
			&=\gamma \left(\left( \frac{1}{s-1}\sum_{r=1}^{s-1}\mathbbm 1_{\vec a_r=b} - \frac{1}{t-1}\sum_{r=1}^{t-1}\mathbbm 1_{\vec b_r=b}\right)u(b)\right.\\
			&\qquad\qquad \left. +\left(\frac{1}{t-1}\sum_{r=1}^{t-1}\mathbbm 1_{\vec a_r=a} - \frac{1}{s-1}\sum_{r=1}^{s-1}\mathbbm 1_{\vec b_r=a} \right)u(a)
			\right)\\
			&\geq \frac{\gamma}{(s-1)(t-1)}\left(\left(t-s\right) \sum_{r=1}^{t-1}\mathbbm 1_{\vec a_r=b}u(b) + \left(s-t\right)\sum_{r=1}^{t-1}\mathbbm 1_{\vec a_r=a}u(a)\right)\\
			&= \frac{\gamma\left(s-t\right)}{s-1}\left(\f\left(a\mid \vec a^{t-1}\right) u(a) - \f\left(b\mid \vec a^{t-1}\right) u(b)\right)
		\end{align*}
		as required.
	\end{proof}
	
	\noindent Lemma \ref{lem:swap_general} provides a sufficient condition to increase all future payoffs by swapping the position of an action $b\in A$ with its next previous occurrence of $a\in A$ within $\vec a$.
	Such a switch increases the average payoff if
	\begin{align*}
		\f\left(a\mid \vec a^{t-1}\right) u(a) - \f\left(b\mid \vec a^{t-1}\right) u(b) >0.
	\end{align*}
	Intuitively, actions with high basic utility will be shifted towards the back, so that they can be played with small punishment.

	\begin{theorem}
		Let $A$ be a finite set of actions.
		Then $\overline{V}_{\gamma}>V^*_{\gamma}$ if and only if for the optimal stationary frequency $\varphi\in\Delta(A)$ there are two actions $a,b\in A$ with $\varphi(a),\varphi(b)>0$ and $u(a)\neq u(b)$.
	\end{theorem}
	
	\begin{proof}
		Necessity is clear. We show the sufficiency.
		Let $\vec a$ be an optimal stationary strategy, write $\varphi(a)$ for $\varphi\left(a\middle\vert\vec a\right)$, and denote by $A^*$ be the set of actions $a\in A$ with $\varphi\left(a\right)>0$.
		By Proposition \ref{pro:stat},
		\begin{align*}
			\varphi\left(a\right)u(a) = \frac{2\gamma - \LV A^*\RV + \sum_{b\in A^*}\frac{u(a)}{u(b)}}{2\gamma \sum_{b\in A^*}\frac{1}{u(b)}}=\frac{u(a)}{2\gamma}-\frac{\LV A\RV -2\gamma}{2\gamma \sum_{b\in A^*}\frac{1}{u(b)}},
		\end{align*}
		which implies that $\varphi(a)\geq\varphi(b)$ if and only if $u(a)\geq u(b)$.
		Moreover, we have
		\begin{align}\label{equ:phi_u_difference}
			\varphi(a)u(a)-\varphi(b)u(b) = \frac{u(a)-u(b)}{2\gamma}
		\end{align}
		for all $a,b\in A^*$.
		As $V^*_{\gamma}$ and $\overline V_{\gamma}$ depend continuously on $\gamma$ and $\left\{u(a)\right\}_{a\in A}$, and $\varphi\left(\cdot\right)\in\mathbb Q^A$ if $\gamma, u(a)\in\mathbb Q$ for all $a\in A$ by Proposition \ref{pro:stat}, we can assume without loss of generality that $\varphi\left(\cdot\right)\in\mathbb Q^A$.
		Thus, there are integers $m_a\in \mathbb N$ for all $a\in A$ such that $\varphi\left(a\right)=\frac{m_a}{m}$, where $m=\sum_{a\in A^*}m_a$.
		Again without loss of generality we can assume that $\vec a$ is the infinite repetition of a sequence of length $m$ in which each action $a\in A^*$ is played exactly $m_a$ times.
		Let $\underline a\in\argmin_{a\in A^*}u(a)$ and $\overline a\in\argmax_{a\in A^*}u(a)$.
		By the premise of the theorem, $u\left(\underline a\right)<u\left(\overline a\right)$, so that $m_{\overline a}>m_{\underline a}$.
		\medskip
		
		\noindent\textbf{Claim 1:} \emph{For all $t\geq 2$ and all $a\in A^*$ it holds that}
		\begin{align*}
			\varphi\left(a\right) - \frac{m}{t-1}\leq\varphi\left(a\middle\vert\vec a^{t-1}\right)\leq\varphi\left(a\right)+ \frac{m}{t-1}.
		\end{align*}
		\begin{proof}
			Let $t\geq 1$.
			There is $k\in\mathbb N$ such that $km\leq t\leq (k+1)m$.
			At $t$, $a$ was chosen at least $km_a$ times, but no more than $km_a+(t-km)$ times.
			Thus,
			\begin{align*}
				\varphi\left(a\mid\vec a^t\right)\geq \frac{km_a}{t}=\frac{km}{t}\varphi(a)=\varphi(a)-\frac{t-km}{t}\geq\varphi(a)-\frac{(k+1)m-km}{t} =\varphi(a)-\frac{m}{t}
			\end{align*}
			and
			\begin{align*}
				\varphi\left(a\mid\vec a^t\right)\leq \frac{km_a + t-km}{t} =\frac{km}{t}\varphi(a) + \frac{t-km}{t}\leq\varphi(a)+\frac{(k+1)m-km}{t} =\varphi(a)+\frac{m}{t}.
			\end{align*}
			Shifting from $t$ to $t-1$ completes the proof.
		\end{proof}

		\noindent Define the following constants
		\begin{align*}
			\delta &= \frac{u\left(\overline a\right)-u\left(\underline a\right)}{4\gamma},\\
			q' &= \frac{\varphi\left(\underline a\right)\left(u\left(\underline a\right) + u\left(\overline a\right)\right)}{\delta +\varphi\left(\underline a\right)\left(u\left(\underline a\right) + u\left(\underline b\right)\right)},\\
			q &= \max\left(\frac{3}{4},q'\right),\\
			\eta &= \frac{1-q}{64}\varphi\left(\underline a\right)\gamma\delta,
		\end{align*}
		and observe that $\delta>0$, so that $q<1$ and $\eta>0$.
		Let
		\begin{align*}
			\varepsilon \leq \min\left(\frac{\delta}{u\left(\underline a\right) + u\left(\overline a\right), \frac{1}{2}\eta}\right)
		\end{align*}
		and let
		\begin{align*}
			T_1 \geq \frac{2m-1 + (m+1)\left(\delta +\varphi\left(\underline a\right)\left(u\left(\underline a\right) + u\left(\overline a\right)\right)\right)}{\delta}
		\end{align*}
		be a multiple of $m$ and be such that $\LV\varphi\left(a\middle\vert\vec a^t\right)-\varphi(a)\RV\leq\varepsilon$ for all $a\in A$ and all $t\geq T_1$.
        Finally, let
		\begin{align*}
			T^*\geq\max\left(2T_1,T_1+4m, T_1+\frac{4\left(\varphi\left(\underline a\right)+2m\right)}{\left(1-q\right)\varphi\left(\underline a\right)}\right)
		\end{align*}
		be such that $\LV U^T_{\gamma}\left(\vec a\right) - V^*_{\gamma}\RV \leq \eta$ for all $T\geq T^*$.
		We show that there is an infinite sequence of $T$'s with $v^T\geq V^*+\eta$ for all $T\geq T^*$.
		For this purpose we will construct for any $T$ in this sequence a history $\vec b$ with $U^T_{\gamma}\left(\vec b\right)\geq V^*+ \eta$.
		By Proposition \ref{pro:v_T_limit}, this is sufficient to prove the theorem.
		
		So, let $T\geq T^*$ be a multiple of $m$.
		We construct $\vec b$ by iteratively switching actions $\underline a$ and $\overline a$ between $T_1+1$ and $T$.
		Specifically, let $\vec c\in\A ^{\infty}$ be a history that has been reached after some switches, and let $T_1+1\leq s\leq T$ be the first occurrence of $\underline a$ such that the period of the last previous occurrence of $\overline a$, denoted by $t<s$, satisfies
		\begin{align}\label{equ:beneficial_switch}
			\f\left(\overline a\middle\vert \vec c^{t-1}\right) u\left(\overline a\right) - \f\left(\underline a\middle\vert \vec c^{t-1}\right) u\left(\underline a\right)\geq\delta.
		\end{align}
		If such $t,s\leq T$ do not exist, let $\vec b=\vec c$.
		Otherwise, note that (\ref{equ:beneficial_switch}) does not depend on $s$, so that the minimality of $s$ implies that there are no instances of $\underline a$ between $t+1$ and $s-1$.
		Let $\vec d$ be such that $\vec d_t=\underline a$, $\vec d_s=\overline a$, and $\vec d_{t'}=\vec c_{t'}$ for all $t'\neq t,s$.
		By Lemma \ref{lem:swap_general},
		\begin{align*}
			U^T_{\gamma}\left(\vec d\right) - U^T_{\gamma}\left(\vec c\right) \geq \frac{\gamma\left(s-t\right)}{\left(s-1\right)T}\delta.
		\end{align*}
		Thus, we say that the switch of $\overline a$ and $\underline a$ in $\vec c$ is \emph{beneficial}.
		Repeat the procedure with $\vec d$ and continue as long as possible.
		Note that all beneficial switches will shift occurrences of $\underline a$ towards the beginning, i.e., $T_1$, and occurrences of $\overline a$ towards the end, i.e., $T$.
		Thus, for each finite $T\geq T^*$ there is a finite number of beneficial switches, so $\vec b$ is well-defined.
		Moreover, for any $T_1+1\leq t\leq T-m$ it holds that the sequence $\big(\vec b_{t+1},\ldots,\vec b_{t+m}\big)$ contains exactly $m_{\underline a}+m_{\overline a}$ periods in which either $\underline a$ or $\overline a$ are being chosen.
		\medskip
		
		\noindent\textbf{Claim 2:} \emph{In history $\vec b$, the last occurrence of $\underline a$ until $T$ is at some $s\leq T_1 + \left(T-T_1\right)q$.}	
		
		\begin{proof}
			Suppose first that in history $\vec b$, there is between $T_1+1$ and $T$ no occurrence of $\overline a$ before $\underline a$.
			Let $s$ be the last occurrence of $\bar a$ and let $s^*$ be the largest multiple of $m$ with $s^*\leq s\leq T$.
			As each occurrence of $\overline a$ between $T_1+1$ and $s^*$ has been replaced by an occurrence of $\underline a$ that originally occurred after $s^*$, it holds that $\varphi\left(\underline a\right)\left(T-s^*\right) \geq \varphi\left(\overline a\right)\left(s^*-T_1\right)$.
			Thus,
			\begin{align*}
				2\varphi\left(\underline a\right)\left(s^*-T_1\right) \leq\left(\varphi\left(\underline a\right)+\varphi\left(\overline a\right)\right)\left(s^*-T_1\right) \leq \varphi\left(\underline a\right)\left(T-T_1\right),
			\end{align*}
			which means that $s^*\leq\frac{1}{2}\left(T_1+T\right)$.
			Hence, since by construction $T\geq T^*\geq T_1+4m$, we find that
			\begin{align*}
				s \leq s^*+m\leq \frac{1}{2}\left(T_1+T\right) + \frac{1}{4}\left(T-T_1\right) \leq T_1+q\left(T-T_1\right),
			\end{align*}
			as required.
			
			Suppose next that after all beneficial switches have been made there is at least one occurrence of $\overline a$ before the last occurrence of $\underline a$.
			Let $t$ be the period of said occurrence of $\overline a$.
			Then, since by the definition of $\vec b$ the switch of $\overline a$ with the last occurrence of $\underline a$ is not beneficial, we have
			\begin{align}\label{equ:switch_not_ben}
				\delta\geq \varphi\big(\overline a\big\vert \vec b^{t-1}\big)u\left(\overline a\right)-\varphi\big(\underline a\big\vert\vec b^{t-1}\big)u\left(\underline a\right).
			\end{align}
			As there is only one occurrence of $\underline a$ in $\vec b$ between $t$ and $T$, we have that $(t-1)\varphi\big(\underline a\big\vert\vec b^{t-1}\big) = T\varphi\left(\underline a\right) -1$, so that
			\begin{align}\label{equ:varphi_a_final}
				\varphi\big(\underline a\big|\vec b^{t-1}\big) = \frac{T}{t-1}\varphi\left(\underline a\right) -\frac{1}{t-1}.
			\end{align}
			Similarly,
			\begin{align*}
				\left(t-1\right)\varphi\big(\overline a\big\vert\vec b^{t-1}\big) = (t-1)\varphi\left(\overline a\middle\vert\vec a^{t-1}\right) - \left(\left(T\varphi\left(\underline a\right) -1\right) - (t-1)\varphi\left(\underline a\middle\vert\vec a^{t-1}\right)\right),
			\end{align*}
			where the expression in the round brackets describes the number of occurrences of $\overline a$ that originally lay between $T_1+1$ and $t$ but have been switched away for some $\underline a$ that originally occurred after $t$.
			Using the bounds that we derived in Claim 1, we find that
			\begin{align*}
				\left(t-1\right)\varphi\big(\overline a\big\vert\vec b^{t-1}\big)
				&\geq (t-1)\left(\varphi\left(\overline a\right)-\frac{m}{t-1}\right) - \left(T\varphi\left(\underline a\right) -1 - (t-1)\left(\varphi\left(\underline a\right)-\frac{m}{t-1}\right)\right)\\
				&= (t-1)\varphi\left(\overline b\right) - \left(T-(t-1)\right)\varphi\left(\underline a\right) - \left(2m-1\right)
			\end{align*}
			Therefore
			\begin{align*}
				\varphi\big(\overline a\big|\vec b^{t-1}\big)\geq \varphi\left(\overline a\right) - \frac{T-(t-1)}{t-1}\varphi\left(\underline a\right)-\frac{2m-1}{t-1}.
			\end{align*}
			This, together with (\ref{equ:phi_u_difference}), (\ref{equ:switch_not_ben}), and (\ref{equ:varphi_a_final}) shows that
			\begin{align*}
				\delta &\geq \left(\varphi\left(\overline a\right) - \frac{T-(t-1)}{t-1}\varphi\left(\underline a\right)-\frac{2m-1}{t-1}\right)u\left(\overline a\right) - \left(\frac{T}{t-1}\varphi\left(\underline a\right) -\frac{1}{t-1}\right)u\left(\underline a\right) \\
				&= \varphi\left(\overline a\right)u\left(\overline a\right) - \varphi\left(\underline a\right)u\left(\underline a\right) - \frac{T-(t-1)}{t-1}\varphi\left(\underline a\right) \left(u\left(\overline a\right) + u\left(\underline a\right)\right) -\frac{2m-1}{t-1}u\left(\overline a\right)\\
				&= \frac{u\left(\overline a\right)-u\left(\underline a\right)}{2\gamma} -  \varphi\left(\underline a\right)u\left(\underline a\right) - \frac{T-(t-1)}{t-1}\varphi\left(\underline a\right) \left(u\left(\overline a\right) + u\left(\underline a\right)\right) -\frac{2m-1}{t-1}u\left(\overline a\right)\\
				&= 2\delta -  \frac{T-(t-1)}{t-1}\varphi\left(\underline a\right) \left(u\left(\overline a\right) + u\left(\underline a\right)\right) -\frac{2m-1}{t-1}u\left(\overline a\right).
			\end{align*}
			Thus,
			\begin{align*}
				\delta\leq  \frac{T-(t-1)}{t-1}\varphi\left(\underline a\right) \left(u\left(\overline a\right) + u\left(\underline a\right)\right)  +\frac{2m-1}{t-1}u\left(\overline a\right)
			\end{align*}
			and solving for $t$ delivers
			\begin{align*}
				t &\leq T\frac{\varphi\left(\underline a\right) \left(u\left(\overline a\right) + u\left(\underline a\right)\right)}{\delta + \varphi\left(\underline a\right) \left(u\left(\overline a\right) + u\left(\underline a\right)\right)} + \frac{2m -1}{\delta + \varphi\left(\underline a\right) \left(u\left(\overline a\right) + u\left(\underline a\right)\right)}+1\\
				&= Tq' + \frac{2m -1}{\delta + \varphi\left(\underline a\right) \left(u\left(\overline a\right) + u\left(\underline a\right)\right)}+1.
			\end{align*}
			Let $s$ be the period of the last occurrence of $\underline a$ in $\vec b$.
			Then $s\leq t+m$.
			Indeed, the sequence $\big(\vec b_{t+1},\ldots,\vec b_{t+m}\big)$ contains at least $m_{\underline a}+m_{\overline a}$ periods in which either $\underline a$ or $\overline a$ is chosen, and at the first such period $\underline a$ is chosen by construction.
			Therefore,
			\begin{align*}
				s &\leq Tq'+ \frac{2m -1}{\delta + \varphi\left(\underline a\right) \left(u\left(\overline a\right) + u\left(\underline a\right)\right)} +1+m\\
				&= \left(T-T_1\right) q' + T_1 -\left(1-q'\right)T_1 + \frac{2m -1}{\delta + \varphi\left(\underline a\right) \left(u\left(\overline a\right) + u\left(\underline a\right)\right)}+1+m\\
				&= \left(T-T_1\right) q' + T_1 + \frac{-\delta T_1 + 2m -1 + \left(m+1\right)\left(\delta + \varphi\left(\underline a\right) \left(u\left(\overline a\right) + u\left(\underline a\right)\right)\right)}{\delta + \varphi\left(\underline a\right) \left(u\left(\overline a\right) + u\left(\underline a\right)\right)}\\
				&\leq \left(T-T_1\right) q + T_1,
			\end{align*}
			where in the last step we use the lower bound for $T_1$ and $q\geq q'$. This concludes the proof of the claim.
		\end{proof}
		
		\noindent So, in history $\vec b$ there are no occurrences of $\underline a$ between $T_1+q\left(T-T_1\right)$ and $T$.
		Let $t^*$ be the smallest integer such that $t^*\geq T_1 + \frac{1+q}{2}\left(T-T_1\right)$ and let $k$ be the number of occurrences of $\underline a$ in $\vec a$ between $t^*+1$ and $T$.
		Then, with the bounds in Claim 1, and since $T\geq T^*\geq T_1+\frac{4\left(\varphi\left(\underline a\right)+m\right)}{\left(1-q\right)\varphi\left(\underline a\right)}$ by construction,
		\begin{align}
			k &= \varphi\left(\underline a\mid\vec a^T\right)T - \varphi\left(\underline a\mid\vec a^{t^*}\right)t^*\notag\\
			&\geq \varphi\left(\underline a\right) T - \varphi\left(\underline a\right)t^* -2m\notag\\
			&\geq \varphi\left(\underline a\right)\left( T- \left( T_1 + \frac{1+q}{2}\left(T-T_1\right) +1\right)\right) -2m\notag\\
			&=  \varphi\left(\underline a\right)\left(T-T_1\right)\frac{1-q}{2} - \varphi\left(\underline a\right) -m\notag\\
			&= \varphi\left(\underline a\right)\left(T-T_1\right)\frac{1-q}{4} + \varphi\left(\underline a\right)\left(T-T_1\right)\frac{1-q}{4}  - \varphi\left(\underline a\right) -2m\notag\\
			&\geq \varphi\left(\underline a\right)\left(T-T_1\right)\frac{1-q}{4} + \varphi\left(\underline a\right)\left(T_1+\frac{4\left(\varphi\left(\underline a\right)+2m\right)}{\left(1-q\right)\varphi\left(\underline a\right)}-T_1\right)\frac{1-q}{4}  - \varphi\left(\underline a\right) -2m\notag\\
			&=\varphi\left(\underline a\right)\left(T-T_1\right)\frac{1-q}{4}.\label{equ:thm:limsup_lower_bound_k}
		\end{align}
		Let $s^1,\ldots, s^k$ be the times of all occurrences of $\underline a$ in $\vec a$ with $T_1 + \frac{1+q}{2}\left(T-T_1\right)\leq s_1\leq\cdots\leq s_k\leq T$.
		Let $t_1\leq \cdots \leq t_k$ be the last $k$ occurrences of $\underline a$ in $\vec b$, and recall that $t_\ell\leq T_1+q\left(T-T_1\right)$ for all $\ell=1,\ldots, k$.
		Thus, $s^\ell-t^\ell\geq \frac{1+q}{2}\left(T-T_1\right)$ for all $\ell=1,\ldots, k$.
		Define histories $\vec b(0),\ldots,\vec b(k)$ as follows.
		Let $\vec b(0)=\vec b$ and for all $\ell=1,\ldots, k$, let $\vec b_{t^{\ell}}\left(\ell\right) = \vec b_{s^{\ell}}\left(\ell -1\right)$, $\vec b_{s^{\ell}}\left(\ell\right) = \vec b_{t^{\ell}}\left(\ell -1\right)$, and $\vec b_t\left(\ell\right)=\vec b_t\left(\ell -1\right)$ for all $t\neq t^{\ell},s^{\ell}$.
		Using Lemma \ref{lem:swap_general} and the fact that $\frac{T-T_1}{T}\geq\frac{1}{2}$ by construction, we therefore have
		\begin{align*}
			U^T_{\gamma}\left(\vec b\left(\ell\right)\right) - U^T_{\gamma}\left(\vec b\left(\ell +1\right)\right)
			&\geq \gamma \frac{1}{T}\frac{s^\ell-t^\ell}{s^\ell-1}\left(\varphi\big(\overline a\big| \vec b^{t^\ell-1}\big)u\left(\overline a\right)-\varphi\big(\underline a\big|\vec b^{t^\ell-1}\big)u\left(\underline a\right)\right)\\
			&\geq \frac{\gamma}{T}\frac{(1+q)\left(T-T_1\right)}{2T}\delta\\
			&\geq \frac{\gamma}{4T}\delta
		\end{align*}
		for all $\ell=0,\ldots,k-1$.
		Observe that the iterative procedure that we used to construct $\vec b$ must have passed through these histories and, in particular, through $\vec b(k)$.
		Thus, with the lower bound in (\ref{equ:thm:limsup_lower_bound_k}) for $k$ we have
		\begin{align*}
			U_{\gamma}^T\left(\vec b\right) - U^T_{\gamma}\left(\vec a\right)
			&\geq U_{\gamma}^T\left(\vec b\right) - U^T_{\gamma}\left(\vec b(k)\right) \\
			&= \sum_{\ell=0}^{k-1} \left(U^T_{\gamma}\left(\vec b(\ell)\right)-U^T_{\gamma}\left(\vec b(\ell +1)\right)\right)\\
			&\geq k\frac{\gamma}{4T}\delta\\
			&\geq \varphi\left(\underline a\right)\frac{1-q}{4}\frac{T-T_1}{T}\frac{\gamma}{4}\delta\\
			&\geq \frac{1-q}{32}\gamma\delta\varphi\left(\bar a\right)\\
			&=2\eta.
		\end{align*}
		Thus,
		\begin{align*}
			v^T_{\gamma} &\geq U^T_{\gamma}\left(\vec b\right) \geq U^T_{\gamma}\left(\vec a\right) +2\eta \geq V^*_{\gamma}-\eta + 2\eta = V^*_{\gamma} + \eta
		\end{align*}
		for all sufficiently large $T$ that are multiples of $m$.
		In particular, $\overline V_{\gamma}=\lim_Tv_{\gamma}^T\geq V^*_{\gamma} + \eta$.
	\end{proof}

	\section{Discounting}\label{sec: discounting}

	In the context of a taste for variety, discounting can have two meanings.
	The first is the classical discounting of future payoffs. This means that we value future positive payoffs less than present ones. For example, we would rather have a delicious meal today than a week from now. The second is a discount on the effect of past uses of actions. As before, this means that the more we experience something, the less we enjoy it. For example, if we eat the same meal every day, we will eventually get tired of it. However, if we had a delicious meal yesterday, we would prefer the same meal today less than if we had it only a year ago.
	
	To take the second meaning into account we define for a discount factor $\lambda\in (0,1)$ the \emph{discounted frequency} of $a$ in the  history $\vec a^{t-1}$ as
	
	\begin{align*}
		\f^{\lambda}\left(a\middle\vert \vec a^{t-1}\right) =\begin{cases}
			\frac{1-\lam}{1-\lam^{t-1}}\sum_{s=1}^{t-1}\lam^{t-s-1}\1_{\vec a_s=a}, & \text{if } t\geq 2,\\
			0, & \text{if } t=1.
		\end{cases}
	\end{align*}
	
	
	 \noindent The fatigue parameter $\g$ is fixed throughout,  and we do not append it in notations.  The utility derived from $\vec{a}$ is defined as,
	\begin{align*}
		U^{\lam, \d}\left(\vec{a}\right)=(1-\d)\sum_{t=1}^{\infty}\d^{t-1} \left(1-\g \f^{\lambda}\left(a_t\middle\vert  \vec a^{t-1}\right)\right)u\left(a_t\right),	\end{align*}
  where $\d>0$ is the future discount factor.

	Let
	\begin{align*}
		V^{\lam, \d} &= \max_{\vec{a}}\, U^{\lam,\d}\left(\vec{a}\right).
	\end{align*}	
	The maximum exists since $U^{\lam,\d}\left(\cdot\right)$ is a continuous function defined on the compact set consisting of all infinite histories. 

	\begin{theorem}\label{thm:discounting}
        There is a function $\d(\lam)< 1$ s.t. for every $\varepsilon>0$ there is $\lam_0$ satisfying
        \begin{align*}
            V^*\leq V^{\lam, \d}<  V^*+\varepsilon,
        \end{align*}
        for every $\lam>\lam_0$ and $\d> \d(\lam)$.
    \end{theorem}
	
	\noindent The theorem states that the stationary strategy is optimal for increasingly patient decision makers. This is another case where cyclical consumption is optimal.
	
	\begin{proof}
		Clearly,  $ V^{\lam,\d }\ge  V^*$ for sufficiently large $\lam$ and $\d $.  This is so, because $  U^{\lam, \d}\left(\vec{a}\right)= V^*$ for a stationary history $\vec{a}$ that achieves  $V^* $.
		We show the inverse direction. Let $\vec{a}$ be an arbitrary sequence and let $\varepsilon>0$. We show that for sufficiently large $\lam$ and $\d$, \begin{equation}\label{eq: at least}
		    U^{\lam, \d}\left(\vec{a}\right)< V^* +\varepsilon.
		\end{equation}

		\noindent To simplify the proof, we use the notation $ \f^{t-1,\lam} \left(a\right)=\f^{\lambda} \big(a\big\vert \vec a^{t-1}\big)$.
		Note that for every $t\ge 2$,   $\sum_a \f^{t-1,\lam} (a)=1 $.
        We let $ \f^{t-1,\lam} $ be the probability distribution that assigns  probability $ \f^{t-1,\lam} \left(a\right)$ to $a$.
        Denote also $\eta^t =(1-\d)\d^{t-1}$ and $\beta^t= \frac{1-\lam}{1-\lam^{t}}$.
        Finally, let $\mathbf 1^t$ stand for the unit $A$-dimensional vector assigning 1 to $a$ when $a_t=a $, and 0 to all other members of $A$.
		Using these notations and the inner product introduced in (\ref{equ:inner_product}), we have
		\begin{align}\label{eq: U lambda}
			U^{\lam,\d}\left(\vec{a}\right)= \sum_{t=1}^{\infty}\eta^t \big\langle \mathbf 1- \g\f^{t-1,\lam},\mathbf 1^t\big\rangle.
		\end{align}
		Denote
		\begin{align*}
			H:= H\left(\vec{a}\right)= \sum_{t=1}^{\infty} \frac{\eta^t}{\beta^t}\big\| \mathbf 1-\g\f^{t, \lam}\big\|^2.
		\end{align*}
		One easily verifies that $ \f^{t,\lam}= \f^{t-1,\lam}+\beta^t \left(\mathbf 1^t-\f^{t-1,\lam}\right)$.
		Hence, with $\epsilon_1 = (1-\delta)\big\|\mathbf 1-\gamma\varphi^{1,\lambda}\big\|$, one obtains
		\begin{align*}
			H &=
            \sum_{t=1}^{\infty} \frac{\eta^t}{\beta^t}\big\| \mathbf 1-\g\f^{t, \lam}\big\|^2\\
			&=
            \epsilon_1 +\sum_{t=2}^{\infty} \frac{\eta^t}{\beta^t}\big\|\big(\mathbf 1-\g\big(\f^{t-1,\lam} \big)\big) -\g \beta^t \big(\mathbf 1^t-\f^{t-1,\lam} \big)\big\|^2  \\
            &=
            \epsilon_1+
			\sum_{t=2}^{\infty} \frac{\eta^t}{\beta^t} \big\| \mathbf 1-\g\f^{t-1, \lam}\big\|^2 - 2\sum_{t=2}^{\infty} \frac{\eta^t}{\beta^t}\big\langle \mathbf 1- \g\f^{t-1,\lam} ,\g \beta^t \big(\mathbf 1^t-\f^{t-1,\lam} \big) \big\rangle \\
			&
            \qquad\qquad\qquad + \g^2\sum_{t=2}^{\infty} \frac{\eta^t}{\beta^t}\big\| \beta^t \big(\mathbf 1^t-\f^{t-1,\lam}\big)  \big\|^2 \\
			&=
            \epsilon_1+
            \sum_{t=2}^{\infty} \left(\frac{\eta^t}{\beta^t}- \frac{\eta^{t-1}}{\beta^{t-1}}\right) \big\| \mathbf 1-\g\f^{t-1, \lam}\big\|^2 +\sum_{t=2}^{\infty}  \frac{\eta^{t-1}}{\beta^{t-1}} \big\| \mathbf 1-\g\f^{t-1, \lam}\big\|^2  \\
			&
            \qquad\qquad\qquad -
			2\g\sum_{t=2}^{\infty}\eta^t\big \langle \mathbf 1- \g\f^{t-1,\lam} , \mathbf 1^t\big \rangle
			+2\g\sum_{t=2}^{\infty}\eta^t\big \langle \mathbf 1- \g\f^{t-1,\lam} , \f^{t-1,\lam} \big \rangle  \\
			&
			\qquad\qquad\qquad + \g^2\sum_{t=2}^{\infty}  \eta^t \beta^t \big  \|   \mathbf 1^t-\g\f^{t-1,\lam} \big  \|^2.
		\end{align*}
		Since $\sum_{t=2}^{\infty}  \frac{\eta^{t-1}}{\beta^{t-1}} \big \| \mathbf 1-\g\f^{t-1, \lam}\big \|^2 =H$, we obtain after rearranging the following key equation:
		\small
		\begin{align}\label{eq: key discount eq}
			2\g\sum_{t=2}^{\infty}\eta^t\big \langle \mathbf 1- \g\f^{t-1,\lam} , \mathbf 1^t\big \rangle
            & =\epsilon_1+ \sum_{t=2}^{\infty} \left(\frac{\eta^t}{\beta^t}- \frac{\eta^{t-1}}{\beta^{t-1}}\right) \big \| \mathbf 1-\g\f^{t-1, \lam}\big \|^2
			 \nonumber +
			2\g\sum_{t=2}^{\infty}\eta^t\big \langle \mathbf 1- \g\f^{t-1,\lam} , \f^{t-1,\lam} \big \rangle \notag\\
			&
            \qquad+ \g^2
            \sum_{t=2}^{\infty}  \eta^t \beta^t \big\| \mathbf 1^t-\g\f^{t-1,\lam} \big\|^2.
		\end{align}
        \normalsize

        \noindent Note that, by the definition of $U^{\lambda,\delta}$,
        \begin{align*}
            2\gamma\sum_{t=2}^{\infty}\eta^t\big \langle \mathbf 1- \g\f^{t-1,\lam} , \mathbf 1^t\big \rangle
            &= 2\gamma U^{\lambda,\delta}\left(\vec a\right) - 2\gamma(1-\delta) u\left(\vec a_1\right) .
        \end{align*}
        Thus, there is $\delta_0$ such that for all $\delta>\delta_0$ we have
        \begin{align}\label{eq: last}
            2\gamma U^{\lambda,\delta}\left(\vec a\right) \leq 2\gamma\sum_{t=2}^{\infty}\eta^t\big \langle \mathbf 1- \g\f^{t-1,\lam} , \mathbf 1^t\big \rangle + \frac{\varepsilon\gamma}{3}.
        \end{align}

		\noindent We turn to the expressions on the right-hand side of (\ref{eq: key discount eq}).
        As $\epsilon_1=\frac{\eta^1}{\beta^1}\left\| \mathbf 1-\g\f^{1, \lam}\right\|^2=1-\d$,  there is $\d_1$ such that $\epsilon_1 <\varepsilon\g/3$ for all $\d>\d_1$.
         	
        We next show that the first sum on the right-hand side of (\ref{eq: key discount eq}) is, for $\lambda$ and $\delta$ sufficiently close to 1, bounded by the same constant.
        For this purpose note first that $ \left \| \mathbf 1-\g\f^{t-1,\lam} \right \|^2 \leq 1$ uniformly.
        Next, observe that
        \begin{align}\label{eq: second term}
            \sum_{t=2}^{\infty} \left\vert\frac{\eta^t}{\beta^t}- \frac{\eta^{t-1}}{\beta^{t-1}}\right\vert
            &\le \frac{1-\d}{1-\lam} \sum_{t=1}^{\infty} \left\vert\d^t\left(1-\lam^{t+1}\right)-\d^{t-1}\left(1-\lam^{t}\right)\right\vert \nonumber  \\
            &=    \frac{1-\d}{1-\lam}\sum_{t=1}^{\infty}\d^{t-1}\left\vert\d\left(1-\lam^{t+1}\right)-\left(1-\lam^{t}\right)\right\vert.
        \end{align}

        \noindent Let $T$ be the largest integer such that\ $\d(1-\lam^{t+1})\ge (1-\lam^{t})$.
        (One easily checks that $T$ is well defined.)
        The right-hand side of  (\ref{eq: second term}) is then bounded from above by
        \small
        \begin{align}
             & \frac{1-\d}{1-\lam} \sum_{t=1}^{T}\d^{t-1}\left(\d\left(1-\lam^{t+1}\right)-\left(1-\lam^{t}\right)\right) +  \frac{1-\d}{1-\lam} \sum_{t=T+1}^{\infty}\d^{t-1}\left(\left(1-\lam^{t}\right)-\d\left(1-\lam^{t+1}\right)\right) \notag\\
             &\qquad \le \frac{1-\d}{1-\lam} \sum_{t=1}^{T}\d^{t-1}\left(\left(1-\lam^{t+1}\right)-\left(1-\lam^{t}\right)\right) +  \frac{1-\d}{1-\lam} \sum_{t=T+1}^{\infty}\d^{t-1}\left(\left(1-\lam^{t+1}\right)-\d\left(1-\lam^{t+1}\right)\right) \notag\\
             &\qquad\le  (1-\d) \sum_{t=1}^{T}\d^{t-1}\lam^{t} + \frac{1-\d}{1-\lam}  \sum_{t=T+1}^{\infty}\d^{t-1}\left(1-\lam^{t+1}-\d\left(1-\lam^{t+1}\right)\right) \notag \\
             &\qquad\le \lam \frac{1-\d}{1-\lam \d} +  \frac{(1-\d)^2}{1-\lam}  \sum_{t=T+1}^{\infty}\d^{t-1}\left(1-\lam^{t+1}\right)\le  \frac{1-\d}{1-\lam \d}  +  \frac{\left(1-\d\right)^2}{1-\lam}  \sum_{t=T+1}^{\infty}\d^{t-1} \notag\\
             &\qquad\le   \frac{1-\d}{1-\lam \d}  +  \frac{1-\d}{1-\lam }\notag\\
             &\qquad\leq   2 \frac{1-\d}{1-\lam }.\label{eq: second term 2}
        \end{align}
        \normalsize
        \noindent As for any $\lam$ there is a function $\d_2(\lam) $ such that when $\d >\d_2(\lam)$, it holds that $2 \frac{1-\d_1(\lam) }{1-\lam }< \varepsilon\g/3 $, we find with (\ref{eq: second term}) and (\ref{eq: second term 2})
        \begin{align}
            \sum_{t=2}^{\infty} \left(\frac{\eta^t}{\beta^t}- \frac{\eta^{t-1}}{\beta^{t-1}}\right) \big \| 1-\g\f^{t-1, \lam}\big \|^2
            &\leq
            \frac{1-\d}{1-\lam}\sum_{t=1}^{\infty}\d^{t-1}\left\vert\d\left(1-\lam^{t+1}\right)-\left(1-\lam^{t}\right)\right\vert
            \leq
            2 \frac{1-\d}{1-\lam }\notag \\
            &< \frac{\varepsilon\g}{3}\label{eq: second term final}
        \end{align}
        as required.

        As for the second sum on the right-hand side of (\ref{eq: key discount eq}), recall that for every $t$ the vector $ \f^{t-1,\lam} $ is a distribution over $A$, so that by the definition of $V^*$ we have, for every $t$,
		$\left \langle \mathbf 1- \f^{t-1,\lam} , \f^{t-1,\lam} \right \rangle \le V^* $.
        Thus,
        \begin{align} \label{eq: last2}
            2\g \sum_{t=2}^{\infty}\eta^t\big \langle \mathbf 1- \g\f^{t-1,\lam} , \f^{t-1,\lam} \big \rangle \le  2\g V^* \sum_{t=2}^{\infty}\eta^t= 2\g V^* (1-\delta) \sum_{t=2}^{\infty}\delta^{t-1} \leq 2\g V^*.
        \end{align}

        \noindent For the last sum on the right-hand side of (\ref{eq: key discount eq}), first note that
      	\begin{align} \label{eq: product series}
            \sum_{t=2}^{\infty}  \eta^t \beta^t=\sum_{t=2}^{\infty}(1-\d) \d^{t-1}\frac{1-\lam}{1-\lam^{t}}\le \sum_{t=1}^{\infty}(1-\d) \d^{t-1}\frac{1-\lam}{1-\lam^{t}}.
		\end{align}


        \noindent For all $\lambda<1$ there is $t^*$ such that $\frac{1-\lambda}{1-\lambda^t}\leq 1-\lambda + \frac{\varepsilon\gamma}{6}$ for all $t\geq t^*$.
        Indeed, note that $\frac{1-\lambda}{1-\lambda^t} = \left(\sum_{s=0}^{t-1}\lambda^s\right)^{-1}\longrightarrow 1-\lambda$ as $t\rightarrow\infty$.
        Moreover, for each $t^*$ there $\delta'<1$ such that $\sum_{s=1}^{t^*}(1-\delta)\delta^{t-1}\frac{1-\lambda}{1-\lambda^t}\leq\frac{\varepsilon\gamma}{6}$ for all $\delta>\delta'$.
        Hence, for each $\lambda<1$ there is $\delta_3\left(\lambda\right)$ such that by (\ref{eq: product series})
        \begin{align}
            \g^2
            \sum_{t=2}^{\infty}  \eta^t \beta^t \big\| \mathbf 1^t-\g\f^{t-1,\lam} \big\|^2
            &\leq
            \sum_{t=2}^{\infty}  \eta^t \beta^t \notag\\
            &\leq
            \sum_{t=1}^{\infty}(1-\d) \d^{t-1}\frac{1-\lam}{1-\lam^{t}} \notag\\
            &= \sum_{s=1}^{t^*}(1-\d) \d^{t-1}\frac{1-\lam}{1-\lam^{t}} + \sum_{t=t^*+1}^{\infty}(1-\d) \d^{t-1}\frac{1-\lam}{1-\lam^{t}} \notag\\
            &\leq \frac{\varepsilon\gamma}{6} + (1-\delta)\left(1-\lambda +\frac{\varepsilon\gamma}{6}\right) \sum_{t=t^*+1}^{\infty}\delta^t \notag\\
            &\leq 1-\lam +\varepsilon\g/3 \notag\\
            &\leq \frac{2\varepsilon\gamma}{3}, \label{eq: last term final}
        \end{align}
        for all $\lambda\geq 1-\frac{\varepsilon\gamma}{3}$ and $\delta\geq\delta_3\left(\lambda\right)$.


        Hence, from (\ref{eq: key discount eq}), (\ref{eq: second term final}), (\ref{eq: last2}), and (\ref{eq: last term final}) we obtain
        \begin{align*}
            2\gamma U^{\lambda,\delta}\left(\vec a\right) &\leq 2\gamma\sum_{t=2}^{\infty}\eta^t\big \langle \mathbf 1- \g\f^{t-1,\lam} , \mathbf 1^t\big \rangle + \frac{\varepsilon\gamma}{3}
            \leq \frac{\varepsilon\gamma}{3} + \frac{\varepsilon\gamma}{3} + 2\gamma V^* + \frac{2\varepsilon\gamma}{3} + \frac{\varepsilon\gamma}{3}\\
            &= \frac{5\varepsilon\gamma}{3} + 2\gamma V^*.
        \end{align*}
        Dividing by $2\gamma$ yields $U^{\lambda,\delta}\left(\vec a\right)\leq V^* + \varepsilon$, as required.
\end{proof}

	\section{Summary}
	We have discussed a dynamic decision problem in which the decision maker's utility derived from a certain action diminishes with the frequency of its use. In the interesting cases where utility is measured by the limit inferior or discounting is applied, the optimal outcome is achieved by a stationary strategy, meaning that periodic consumption is optimal.

	\printbibliography

\end{document}